\pdfoutput=1
\documentclass[10pt,journal,compsoc]{IEEEtran}
\ifCLASSOPTIONcompsoc
\usepackage[nocompress]{cite}
\else
\usepackage{cite}
\fi

\ifCLASSINFOpdf
\usepackage[pdftex]{graphicx}
\graphicspath{{Figures/pdf/}{Figures/jpeg/}{Figures/png/}{Figures/eps/}}

\DeclareGraphicsExtensions{.pdf,.jpeg,.png,.eps}
\else
\usepackage[dvips]{graphicx}
\graphicspath{{Figures/eps/}}
\DeclareGraphicsExtensions{.eps}
\fi

\usepackage[cmex10]{amsmath}
\usepackage{algorithm,algorithmic}
\usepackage{subfigure}
\usepackage{amssymb}
\usepackage{algorithm} 
\usepackage{algorithmic}

\usepackage{array}
\usepackage{booktabs}
\usepackage{amsthm}
\usepackage{verbatim}

\newtheorem{theorem}{Theorem}
\newtheorem{corollary}{Corollary}
\newtheorem{lemma}{Lemma}
\newtheorem{definition}{Definition}

\usepackage{setspace}
\usepackage{bbding}
\usepackage{pifont}
\usepackage{wasysym}
\usepackage{amssymb}
\usepackage{makecell}
\usepackage{multirow}
\usepackage{bigstrut}
\usepackage{multicol}



\begin{document}
\title{Efficient Strong Privacy-Preserving Conjunctive Keyword Search Over Encrypted Cloud Data}
	
\author{Chang~Xu,~\IEEEmembership{Member,~IEEE,}
	Ruijuan~Wang,
	Liehuang~Zhu,~\IEEEmembership{Member,~IEEE,}
	Chuan~Zhang,~\IEEEmembership{Member,~IEEE,}
	Rongxing Lu,~\IEEEmembership{Fellow,~IEEE,}
	and~Kashif~Sharif,~\IEEEmembership{Senior~Member,~IEEE}
	
	\IEEEcompsocitemizethanks{
		\IEEEcompsocthanksitem Chang Xu, Liehuang Zhu, and Chuan Zhang are with School of Cyberspace Science and Technology, Beijing Institute of Technology, Beijing 100081, China.\protect\\
		E-mail: \{xuchang, liehuangz, chuanz\}@bit.edu.cn
		\IEEEcompsocthanksitem Ruijuan Wang and Kashif Sharif are with School of Computer Science and Technology, Beijing Institute of Technology, Beijing 100081, China.\protect\\
		E-mail: \{ruijuanw, kashif\}@bit.edu.cn
		\IEEEcompsocthanksitem Rongxing Lu is with the Faculty of Computer Science, University of New Brunswick, Fredericton, NB E3B 5A3, Canada.\protect\\
		E-mail: rlu1@unb.ca.}
	}

\markboth{Journal of \LaTeX\ Class Files,~Vol.~x, No.~x,~x}%
{Shell \MakeLowercase{\textit{et al.}}: Efficient Strong Privacy-Preserving Conjunctive Keyword Search Over Encrypted Cloud Data}

\IEEEtitleabstractindextext{%
	\begin{abstract}
	Searchable symmetric encryption (SSE) supports keyword search over outsourced symmetrically encrypted data. Dynamic searchable symmetric encryption (DSSE), a variant of SSE, further enables data updating. Most DSSE works with conjunctive keyword search primarily consider forward and backward privacy. Ideally, the server should only learn the result sets involving all keywords in the conjunction. However, existing schemes suffer from keyword pair result pattern (KPRP) leakage, revealing the partial result sets containing two of query keywords. We propose the first DSSE scheme to address aforementioned concerns that achieves strong privacy-preserving conjunctive keyword search. Specifically, our scheme can maintain forward and backward privacy and eliminate KPRP leakage, offering a higher level of security. The search complexity scales with the number of documents stored in the database in several existing schemes. However, the complexity of our scheme scales with the update frequency of the least frequent keyword in the conjunction, which is much smaller than the size of the entire database. Besides, we devise a least frequent keyword acquisition protocol to reduce frequent interactions between clients. Finally, we analyze the security of our scheme and evaluate its performance theoretically and experimentally. The results show that our scheme has strong privacy preservation and efficiency.

\end{abstract}
	
}
\maketitle
\IEEEdisplaynontitleabstractindextext
\IEEEpeerreviewmaketitle
	
\IEEEraisesectionheading{\section{Introduction}\label{s_1}}
\IEEEPARstart{W}{ith} the advent of cloud computing, outsourcing the storage and processing of large data collection to third-party servers has gained significant popularity. However, data privacy is a primary concern as the third-party cloud server cannot be fully trusted. These privacy issues may cause reputational damage and/or location leakage \cite{Liu2017SocialNetwork, ZHANG2020406}. Generally, users do not desire that their sensitive data be disclosed to an untrusted server. A straightforward solution is to encrypt the data before uploading it to the cloud server. However, this limits the search operations, as encrypted data may not be directly searchable. 

To address these challenges, the searchable encryption (SE) technology is proposed to achieve keyword searching over ciphertext without revealing sensitive data and queries to the cloud server. Searching over encrypted data inevitably reduces search efficiency; hence, the SE schemes try to maintain an acceptable compromise between security and efficiency. In recent years, a series of schemes have been proposed \cite{cash2013highly, yang2017rspp, lai2018result, liu2020privacy}, which enable the data owner to encrypt the data and generate encrypted indices for searching. The cloud server can then retrieve the stored ciphertext based on the received search tokens.

In general, searchable encryption can be divided into two representative techniques: Symmetric searchable encryption (SSE), and Public-key searchable encryption. SSE has been extensively studied, and various schemes \cite{curtmola2011searchable, cash2014dynamic, li2019searchable, ghareh2018new} have been presented successively. However, the majority of schemes are severely limited to support single keyword search. In realistic scenarios, conjunctive keyword search is more appropriate, such as, e-health systems \cite{zhang2017searchable}, task recommendation systems \cite{shu2018privacy, zhang2021location}, and e-mail systems~\cite{wu2019vbtree}. Hence, we focus on conjunctive keyword search in SSE in this work.

A simple solution to support conjunctive keyword search has been described in \cite{JhoH13}. This solution returns the intersection of search results of each single keyword in the conjunction. Thus, this search method is inefficient due to repeated searches in the database. Furthermore, it will breach the searched data's secrecy. The cloud server should be allowed to acquire only the encrypted file identifiers corresponding to all keywords in the conjunctive search query $q=(w_1\wedge\cdots\wedge{w_n})$. To maintain this secrecy, Cash et al. \cite{cash2013highly} proposed the first sub-linear SSE protocol supporting conjunctive keyword search, named Oblivious Cross-Tags (OXT). The earlier SSE constructions \cite{Golle2004} scale linearly with the size of the entire database. Sub-linear means that the search complexity is independent of the total number of documents stored in the database. To reduce search overhead, the complexity of OXT is proportional to the number of matches involving the least frequent keyword as the $s$-term in the conjunction. Subsequently, many conjunctive SSE schemes based on OXT have been proposed \cite{wang2018vsse, kermanshahi2021multi, gan2022vsse}.

\begin{table*}[htbp] 
	\caption{Functionality Comparison of Existing Schemes and This Work.}
	\label{tab1}
	\centering
	\begin{tabular}{|p{2.1cm}<{\centering}|p{2cm}<{\centering}|p{2cm}<{\centering}|p{2cm}<{\centering}|p{2.3cm}<{\centering}|p{2cm}<{\centering}|}
		\hline  
		\makecell{} & \makecell{\textbf{Conjunctive} \\ \textbf{keyword}} & \makecell{\textbf{Forward} \\ \textbf{privacy}} & \makecell{\textbf{Backward} \\ \textbf{privacy}} & \makecell{\textbf{Keyword pair} \\ \textbf{result privacy}} & \makecell{\textbf{Non-}\\\textbf{interactive}}  \\
		\hline
		OXT\cite{cash2013highly} & \checkmark & \ding{55} & \ding{55}  & \ding{55} & \ding{55}\\
		
		\hline
		HXT\cite{lai2018result} & \checkmark & \ding{55} & \ding{55}  & \checkmark & \ding{55}  \\
		
		\hline
		ODXT\cite{patranabis2021forward} & \checkmark & \checkmark & \checkmark  & \ding{55} & \ding{55}\\
		
		\hline
		SE-EPOM\cite{liu2020privacy} & \checkmark & \ding{55} & \ding{55}  & \ding{55} & \checkmark   \\
		
		\hline
		Ours & \checkmark & \checkmark & \checkmark  & \checkmark & \checkmark   \\
		
		\hline  
	\end{tabular} 
\end{table*}

According to Lai et al.’s scheme \cite{lai2018result}, the OXT protocol suffers Keyword Pair Result Pattern (KPRP) leakage during the search phase. This means that for an $n$ keywords conjunctive search query $q$, the server can learn the encrypted file identifiers involving each pair of query keywords $(w_1, w_i), 2\leq{i}\leq{n}$, where the keyword $w_1$ is assumed to be the $s$-term. The file-injection attack \cite{zhang2016all} resorts to using KPRP leakage to expose all keywords in a conjunctive query with 100\% accuracy. To eliminate KPRP leakage, Lai et al. \cite{lai2018result} proposed the Hidden Cross-Tags (HXT) protocol. Unfortunately, all of the above SSE schemes are limited to static databases.

Adding and deleting files is generally required in real-world scenarios to dynamically update the database, which raises security concerns. For instance, the server might intentionally match the previous search tokens with the newly added file indexes to infer the keywords contained in the file and infer the keywords of the query from repeated search queries. Zhang et al. \cite{zhang2016all} presented file-injection attacks. Specifically, an untrusted server first crafts a set of files and tricks the client into encrypting them. After that, the server searches for the injected files using the prior tokens. According to the known keywords, the server can infer which keyword is involved in the token. File-injection attacks are effective to break the privacy of client queries and require only a small number of injected files. This attack was also described as leakage-abuse attacks in known-document and chosen-document attack setting by Cash et al. \cite{cash2015}. File-injection attacks and leakage-abuse attacks are based on the prior knowledge of an adversarial server. File-injection attacks require less prior knowledge and the server must inject files. Therefore, we consider the above attacks. To deal with the attacks, forward privacy and backward privacy were first introduced informally in by Stefanov et al. \cite{stefanov2013practical}, and later formalized by Bost et al. \cite{bost2016forward, bost2017forward}. \textbf{Forward privacy} ensures that the newly updated files cannot be linked with the previously executed search, which prevents the server from inferring the keywords. \textbf{Backward privacy} requires that deleted files cannot be retrieved in subsequent search queries. That is, search queries should not reveal information about files that have already been deleted from the database. 

Most of the published forward and backward private DSSE schemes only support single keyword search. For conjunctive queries, Zuo et al. \cite{zuo2020iacr} proposed FBDSSE-CQ scheme using the extended bitmap index and achieving both forward and backward privacy. However, the number of keywords in a conjunction query is limited. To date, for the existing works that is not limited by the number of search keywords, only the Oblivious Dynamic Cross Tags (ODXT) scheme \cite{patranabis2021forward} supports conjunctive keyword search guaranteeing forward and backward privacy. However, there is still the threat of KPRP leakage.

Besides, the previous OXT-based SSE protocols require the search user to ask the data owner during each search because the least frequent keyword in the conjunction should be known in advance. However, if a search user initiates distinct search queries multiple times, it should interact with the data owner for each query, and such frequent interactions cause additional communication overhead problems. In actuality, the data owner outsources the data to the server in the expectation of not having to do anything else but update the data and distribute secret keys. This work aims to propose a DSSE scheme that guarantees forward and backward privacy and eliminates the KPRP leakage to achieve strong privacy, and enable search users to obtain the least frequent keyword without interacting with the data owner.

\textbf{Our Contributions}. This work develops a solution for the privacy-preserving leakage problem of conjunctive DSSE and proposes an Efficient Strong Privacy-preserving Conjunctive Keyword Search (ESP-CKS) scheme. It ensures forward and backward privacy of DSSE and avoids keyword pair result pattern leakage. Table \ref{tab1} shows a comparison of our scheme with prior representative works for conjunctive keyword search. Our contributions are listed below in detail.

\begin{itemize}
	\item We propose the first strong privacy-preserving conjunctive keyword search scheme named ESP-CKS. It guarantees forward and backward privacy for dynamic databases. Besides, our scheme can prevent KPRP leakage in conjunctive search queries. Although the previous schemes support forward and backward privacy, they suffer from KPRP leakage.
	
	\item Our ESP-CKS scheme's search complexity is independent of the total number of documents in the database, but scales with the update frequency of the least frequent keyword in the conjunction. To address the issue of frequent interactions between the data owner and search users in previous work, we propose the first non-interactive least frequent keyword acquisition protocol.
	
	\item We conduct extensive experiments to evaluate the scheme's performance in terms of storage, computation, and communication. We prove that our scheme is more secure and efficient than the other two conjunctive keyword search approaches.
\end{itemize}

\section{Problem formulation}\label{s_2}
This section describes the system model and then introduces the threat model and design goals.

\subsection{System model}\label{s_2_1}
The system consists of three main entities as shown in Figure~\ref{Fig_1}: a data owner, a cloud server, and search users. 
\begin{itemize}
  \item Data Owner: The data owner is in charge of generating the system parameters and the secret key. The data owner encrypts documents via symmetric encryption and generates related indexes according to system parameters. Then it submits both indexes and encrypted documents to the cloud server. It can also add and delete documents from the encrypted database on the cloud server.
  
  \item Cloud Server: The cloud server holds unlimited computation and storage capacities. It stores documents and the indexes delivered by the data owner and handles search queries.
    
  \item Search Users: A search user is permitted to search the documents stored on the cloud server. It computes search tokens for desired keywords and sends them to the cloud server. After receiving the search results from the cloud server, it decrypts them to obtain the target documents' indexes.
\end{itemize}

We call the data owner and search users as clients, and the cloud servers as the server.

\subsection{Threat Model}\label{s_2_2}
In our schemes, the data owner and search users are assumed to be fully trusted. Specifically, they honestly follow the protocol and never expose its encryption keys to other entities except for permitted search users. The cloud server is considered honest but curious and always honestly executes the protocol but may attempt to infer information about data objects and the content of queries. 

\begin{figure}[htbp]
\centering
\includegraphics[width=3in]{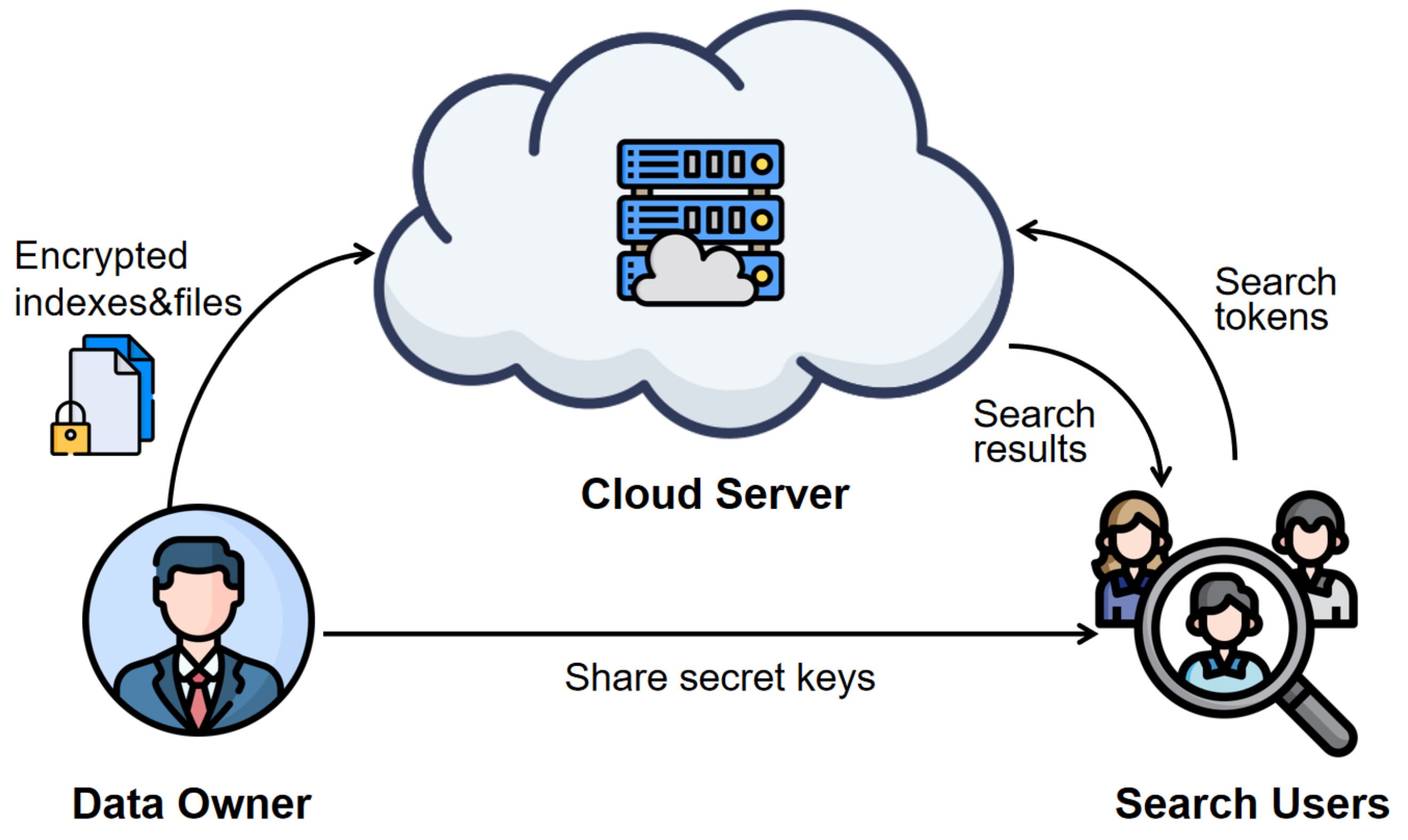}
\caption{System model}
\label{Fig_1}
\end{figure}

\subsection{Design Goals}\label{s_2_3}
We aim to design a strong privacy-preserving conjunctive keyword DSSE scheme, which enables the cloud server to provide storage and search services in a privacy-preserving manner. Our designed scheme should meet the following security requirements and features.

\begin{itemize}
    \item Forward Privacy: The cloud server should not learn whether the newly stored documents contain the previously searched keywords.

    \item Backward Privacy: The cloud server should not learn the deleted document when searching.
    
    \item Keyword Pair Result Privacy: The cloud server should not obtain the partial query result sets for two specific keywords, eliminating KPRP leakage.
    
    \item Non-interactive: The search user should not interact with the data owner for acquiring the keyword with the least frequency in each search query.
    
    \item Search Efficiency: The search complexity should be sub-linear. That means it is independent of the total number of stored documents but is correlated with the frequency of the least frequent keyword in the conjunction.
\end{itemize}

This paper considers the scheme with the above-mentioned privacy properties as a strong privacy-preserving scheme. 

\section{Preliminaries}\label{s_3}
In this section, we first present notations used in this work in Table \ref{notation}. Then we present the syntax of DSSE and a formal security definition. Finally, we give a brief description of Bloom filter and symmetric-key hidden vector encryption scheme, which are adopted to construct our scheme.
\begin{table}[htbp] 
	\centering
	\begin{small}
		\caption{Notations and Descriptions.}
		\label{notation}
		\centering
		\begin{tabular}{|c|l|} 
			\hline  
			\textbf{Notation} & \textbf{Meaning} \\
			\hline
			\hline
			$\lambda$ & security parameter\\
			\hline
			$w$ & a keyword\\
			\hline
			$op$ & operation in $\{add,del\}$ \\
			\hline
			$id_i$ & identifier of $i$-th file where $id_i\in\{0,1\}^\lambda$\\
			\hline
			$W_i$ & list of keywords contained in the file $id_i$\\
			\hline
			$M$ & number of files stored in $\mathsf{DB}$ \\
			\hline
			$\mathsf{DB}$ & database $\{op_i,id_i,W_i\}_{i=1}^{M}$\\ 
			\hline
			$\mathcal{W}$ & set of keywords $\cup_{i=1}^{M}W_i$\\
			\hline
			$|\mathcal{W}|$ & number of keywords stored in $\mathsf{DB}$ \\
			\hline 
			$q$ & conjunctive query $(w_1\wedge\cdots\wedge{w_n})$\\
			\hline
			$\mathcal{Q}$ & list $\{w_1,\cdots,w_n\}$ of keywords for $q$ \\
			\hline
			$\mathsf{DB}(w)$ & file identifiers containing $w$\\ 
			\hline
			$\mathsf{DB}(q)$ & file identifiers $\bigcap_{i=1}^{n}\mathsf{DB}(w_i)$ for $q$\\
			\hline
			$w_s$ & the keyword with the least update frequency \\
			\hline
			$cnt_w$ & the update frequency of $w$\\
			\hline
			$ecnt_w$ & encrypted value of $cnt_w$\\
			\hline
			$\Gamma$ & dictionary of keyword-frequency pairs $(w,cnt_w)$\\
			\hline
			$ftoken_{w}$ & frequency token of $w$\\
			\hline 
			$\mathcal{T}$ & list of frequency tokens for $q$\\
			\hline
			$\Delta$ & dictionary of token-value pairs $(ftoken_{w},ecnt_w)$\\
			\hline
			$BF$ & a Bloom filter\\
			\hline
			$m$ & size of $BF$\\
			\hline
			$m'$ & number of non-wildcard elements in $BF$\\
			\hline
			$N$ & number of triple $(op,id,w)$ stored in $\mathsf{DB}$\\
			\hline
			$[t]$ & set of integers $\{1,2,\cdots{t}\}$\\
			\hline 
			$poly(\lambda)$ & an unspecified polynomial in $\lambda$ \\
			\hline
			$negl(\lambda)$ & a negligible function in $\lambda$\\
			\hline
			$||$ & concatenation of strings\\
			\hline
			$x\stackrel{R}{\leftarrow}\mathcal{X}$ & uniformly sampling a random $x$ from $\mathcal{X}$\\
			\hline 
			\hline
		\end{tabular} 
	\end{small}
\end{table}

\subsection{Syntax of Dynamic Searchable Symmetric Encryption}\label{s_3_1}
A Dynamic Searchable Symmetric Encryption (DSSE) scheme $\Pi=(\textsc{Setup}, \textsc{Update}, \textsc{Search})$ consists of an algorithm \textsc{Setup} run by a client and two protocols \textsc{Update} and \textsc{Search} run by a client and a cloud server.

\begin{itemize}
  \item $\textsc{Setup}(\lambda,\mathsf{DB})$: Given the security parameter $\lambda$ and a database $\mathsf{DB}$, the client executes the algorithm to generate $(sk,\sigma,\mathsf{EDB})$, where $sk$ is a secret key, $\sigma$ is the client's state, and $\mathsf{EDB}$ is an empty encrypted database that is sent to the server.
  
  \item $\textsc{Update}(sk,\sigma,op,id,w;\mathsf{EDB})$: The client takes as input the secret key $sk$, the sate $\sigma$, an operation $op$ which can be $add$ or $del$, a file identifier $id$ and a keyword $w$. The server takes as input $\mathsf{EDB}$. After executing the protocol, the client updates its internal sate $\sigma'$ and the server updates encrypted database $\mathsf{EDB}'$.
  
  \item $\textsc{Search}(sk,\sigma,q;\mathsf{EDB})$: The client inputs the secret key $sk$, the state $\sigma$, and a search query $q$. The server takes as input $\mathsf{EDB}$. After executing the protocol, the client outputs $\mathsf{DB}(q)$ as the search result. 
\end{itemize}

\subsection{Security Definition of DSSE}\label{s_3_2}
The security of a DSSE scheme states that the server must learn as little as possible about the content of the database and queries. We expect that the adversary cannot learn anything beyond certain obvious leakages. The security of a DSSE scheme can be parametrized by a stateful leakage function $\mathcal{L}=(\mathcal{L}_{Setup},\mathcal{L}_{Update},\mathcal{L}_{Search})$ that expresses the information leaked to the adversary throughout each operation. We describe two probabilistic experiments in the real world and the ideal world as follows.

\begin{itemize}
  \item $\textbf{Real}^{\Pi}_{\mathcal{A}}(\lambda,Q)$: The adversary $\mathcal{A}$ chooses a database $\mathsf{DB}$ and a query list $\mathsf{q}$. The experiment runs $\textsc{Setup}(\lambda,\mathsf{DB})$ and returns $\mathsf{EDB}$ to $\mathcal{A}$. Then, for each $i\in{Q}$ ($Q=|\mathsf{q}|$), the experiment responds to the query by running $\textsc{Update}(sk,\sigma_i,\mathsf{q_i};\mathsf{EDB}_i)\rightarrow{(\sigma_{i+1},\mathsf{EDB}_{i+1})}$ or $\textsc{Search}(sk,\sigma_i,\mathsf{q_i};\mathsf{EDB}_i)\rightarrow{\mathsf{DB}(\mathsf{q_i})}$ depending on whether $\mathsf{q_i}$ is an update query or a search query. Eventually, $\mathcal{A}$ outputs a bit $b\in\{0,1\}$.
  
  \item $\textbf{Ideal}^{\Pi}_{\mathcal{A}}(\lambda,Q)$: The adversary $\mathcal{A}$ chooses a database $\mathsf{DB}$. Given the leakage function $\mathcal{L}_{Setup}$, the simulator $\mathcal{S}$ returns $\mathsf{EDB}$ to $\mathcal{A}$. Then, $\mathcal{A}$ adaptively chooses a query list $\mathsf{q}$. For each $i\in{Q}$, $\mathcal{S}$ answers the query $\mathsf{q_i}$. Eventually, the adversary $\mathcal{A}$ outputs a bit $b\in\{0,1\}$.
\end{itemize}

\begin{definition}
A DSSE scheme $\Pi$ with a collection of leakage function $\mathcal{L}$ is adaptively-secure if for any probabilistic polynomial-time adversary $\mathcal{A}$ issuing a maximum of $Q=poly(\lambda)$ queries, there exists a probabilistic polynomial-time simulator $\mathcal{S}$ such that 
\begin{eqnarray*}
|Pr[\textbf{Real}^{\Pi}_{\mathcal{A}}(\lambda,Q)=1]-Pr[\textbf{Ideal}^{\Pi}_{\mathcal{A}}(\lambda,Q)=1]|\leq{negl(\lambda)}.
\end{eqnarray*}
\end{definition}

\subsection{Bloom Filter}\label{s_3_3}
Bloom filter is a space-efficient probabilistic data structure used to represent a set $\mathcal{S}=\{s_1,s_2,...,s_N\}$ of $N$ elements. Its main functionality is to support fast set membership verification. A Bloom filter consists of a binary array of $m$-bits which is initially all 0. It is associated with $k$ independent hash functions $\{H_i\}_{1\leq{i}\leq{k}}$. Given each element $s\in\mathcal{S}$, the bits at positions $\{H_i(s)\}_{1\leq{i}\leq{k}}$ in the Bloom filter are set to 1. To test the existence of $s$, check whether all bits at positions $\{H_i(s)\}_{1\leq{i}\leq{k}}$ are equal to 1. If so, $s\in\mathcal{S}$ with high probability due to the false positive rate. Otherwise, $s\notin\mathcal{S}$ with the probability 1. The false positive means that $s\notin\mathcal{S}$ but membership test returns to 1. Suppose elements in $\mathcal{S}$ are hashed into the Bloom filter, the false positive rate $P_e$ is
\begin{eqnarray*}
P_e \leq (1-e^{-kN/m})^k.
\end{eqnarray*}
$P_e$ can be negligible by choosing optimal parameters $m$ and $k$. Given $N, P_e$, the optimal choice of $k$ is $k\approx{\log_2(1/P_e)}$, while the required $m\approx{1.44\cdot{\log_2(1/P_e)}}\cdot{N}$ \cite{broder2004network}.

\subsection{Symmetric-key Hidden Vector Encryption}\label{s_3_4}
Symmetric-key Hidden Vector Encryption (SHVE) \cite{lai2018result} is a lightweight predicate encryption scheme that supports comparison over encrypted data. Assume that $\Sigma=\{0,1\}$ is a finite set of attributes, and $\ast$ is a wildcard symbol not in $\Sigma$, $\Sigma_{\ast}=\Sigma\cup{\{\ast\}}$. For each index vector $\textbf{\rm{x}}=(x_1,\cdots,x_m)\in{\Sigma^m}$ and predicate vector $\textbf{\rm{v}}=(v_1,\cdots,v_m)\in{\Sigma^m_{\ast}}$, we have:
\begin{gather*}
P_{\textbf{\rm{v}}}^{\rm{SHVE}}(\textbf{\rm{x}})=\left\{
         \begin{array}{lc}
            1, & \forall~{1\leq{i}\leq{m}(v_i=x_i~or~v_i=\ast)} \\
            0, &   otherwise.
         \end{array}
\right.
\end{gather*}

The predicate $P_{\textbf{\rm{v}}}^{\rm{SHVE}}(\textbf{\rm{x}})=1$ means that the vector \textbf{x} matches \textbf{v} in all the locations that are non-wildcard characters. The details of the SHVE are as follows:

\begin{itemize}
    \item ${\rm{SHVE.Setup}}(1^\lambda)\rightarrow{(msk,\mathcal{M})}$: Given the security parameter $\lambda$, the algorithm outputs an uniformly sampled master secret key $msk\stackrel{R}{\leftarrow}\{0,1\}^{\lambda}$, and defines the payload message space $\mathcal{M}= \{``True"\}$.

    \item ${\rm{SHVE.Enc}}(msk,\mu= ``True", \textbf{x}\in{\Sigma^m})\rightarrow\textbf{c}$: The algorithm takes as input the master secret key $msk$, a message $\mu= ``True"$ and an index vector $\textbf{x}=(x_1,\cdots,x_m)\in{\Sigma^m}$. It returns the ciphertext
    \begin{eqnarray*}
    \textbf{c}=\{c_l\}_{l\in{[m]}},
    \end{eqnarray*}
    where for each $l\in{[m]}$, $c_l=F(msk,x_l||l)$.

    \item ${\rm{SHVE.KeyGen}}(msk,\textbf{v}\in{\Sigma^m_{\ast}})\rightarrow\textbf{s}$: The algorithm takes as input the master secret key $msk$ and a predicate vector $\textbf{v}=(v_1,\cdots,v_m)\in{\Sigma^m_{\ast}}$. Let set $S=\{l_j\in{[m]}|v_{l_j}\neq\ast\}$ be the set of all locations in \textbf{v} that do not contain the wildcard characters. The algorithm randomly samples to get $K_0\stackrel{R}{\leftarrow}\{0,1\}^{\lambda+\log\lambda}$, and calculates:
    \begin{gather*}
    d_0=\oplus_{j\in[|S|]}{(F(msk,v_{l_j}||l_j))}\oplus{K_0}, \text{and} \\
    d_1={\rm{Sym.Enc}}(K_0,0^{\lambda+\log\lambda}).
    \end{gather*}
    Here $F$ is a pseudorandom function, and Sym.Enc denotes a symmetric encryption algorithm. The algorithm finally returns the decryption key
    \begin{eqnarray*}
    \textbf{s}=(d_0,d_1,S).
    \end{eqnarray*}
 
    \item ${\rm{SHVE.Query}}(\textbf{s},\textbf{c})\rightarrow{\mu(\perp)}$: On input of the ciphertext \textbf{c} and the key \textbf{s}, the query algorithm parse $\textbf{s}=(d_0,d_1,S)$ and $\textbf{c}=(\{c_l\}_{l\in{[m]}})$, where $S=\{l_1,l_2,\cdots,l_{|S|}\}$. Then, it computes
    \begin{eqnarray*}
    {K_0}'=(\oplus_{j\in[|S|]}{c_{l_j}})\oplus{d_0}.
    \end{eqnarray*}
    Next the decryption algorithm calculates
    \begin{eqnarray*}
    \mu'={\rm{Sym.Dec}}({K_0}',d_1).
    \end{eqnarray*}
    If $\mu'=0^{\lambda+\log\lambda}$, it outputs $``True"$, otherwise outputs $\perp$.
\end{itemize}

\textbf{Correctness}: Given the ciphertext $\textbf{c}=(\{c_l\}_{l\in{[m]}})$ related to index vector $\textbf{\rm{x}}=(x_1,\cdots,x_m)$, the decryption key $\textbf{s}=(d_0,d_1,S)$ corresponding to predicate vector $\textbf{\rm{v}}=(v_1,\cdots,v_m)$, and the set of locations $S=\{l_1,l_2,\cdots,l_{|S|}\}$, we analyze the correctness of aforementioned scheme in the following two scenarios:
\begin{itemize}
    \item If $P_{\textbf{\rm{v}}}^{\rm{SHVE}}(\textbf{\rm{x}})=1$, $v_{l_j}=x_{l_j}$ holds for each $j\in[|S|]$. This means that we can calculate $c_{l_j}=F(msk,v_{l_j}||l_j)$ for each $j\in[|S|]$. Then we have
    \begin{eqnarray*}
    \begin{aligned}
    {K_0}'&=(\oplus_{j\in[|S|]}{c_{l_j}})\oplus{d_0}\\
      &=(\oplus_{j\in[|S|]}{c_{l_j}})\oplus{(\oplus_{j\in[|S|]}{(F(msk,v_{l_j}||l_j))})}\oplus{K_0}\\
      &=(\oplus_{j\in[|S|]}{F(msk,v_{l_j}||l_j)})\\&~~~~~~~~~~~\oplus{(\oplus_{j\in[|S|]}{(F(msk,v_{l_j}||l_j))})}\oplus{K_0}\\
      &=K_0,\text{and}
    \end{aligned}
    \end{eqnarray*}
    \begin{eqnarray*}
    \begin{aligned}
    \mu'&={\rm{Sym.Dec}}({K_0}',d_1)=0^{\lambda+\log\lambda}.
    \end{aligned}
    \end{eqnarray*}
    Finally, the query algorithm SHVE.Query returns $``True"$.
    
    \item If $P_{\textbf{\rm{v}}}^{\rm{SHVE}}(\textbf{\rm{x}})=0$, there exists some $j\in[|S|]$ such that $v_{l_j}\neq{x_{l_j}}$. This implies that $c_{l_j}\neq{F(msk,v_{l_j}||l_j)}$, and hence, ${K_0}'\neq{K}$, so that we have $\mu'\neq0^{\lambda+\log\lambda}$ and the algorithm returns the failure symbol $\perp$.
\end{itemize}

\section{Construction of the Proposed Scheme}\label{s_4}
In this section, we present the detailed construction of our ESP-CKS scheme. We first introduce our least frequent keyword acquisition protocol to achieve non-interactive acquisition of the $s$-term $w_s$ with the least update frequency in the conjunction. After that, we propose the algorithm for the setup, update and search phases, respectively.

\subsection{Least Frequent Keyword Acquisition Protocol}\label{s_4_1}
The existing conjunctive keyword SSE schemes require that a search user asks for the keyword with the least frequency of updates from the data owner in each search query. However, this produces redundant communication overhead. Ideally, after the update is complete, the data owner should not be required to perform additional interactions for each search. We propose the privacy-preserving Least Frequent Keyword Acquisition (LFKA) protocol to address this problem. The protocol guarantees that a search user can obtain the least frequent keyword without frequent interactions with the data owner. Meanwhile, the privacy of update frequency is protected.

Let $F_1:\{0,1\}^{\lambda}\times\{0,1\}^{\ast}\rightarrow\{0,1\}^{2\lambda}$ and $F_2:\{0,1\}^{\lambda}\times\{0,1\}^{\ast}\rightarrow\{0,1\}^{\lambda}$ be pseudorandom functions (PRFs). Our LFKA protocol $\Lambda=(KeyGen, FreqSetup, TokenGen, FreqFind, Compare)$ is constructed as follows.

\begin{itemize}
  \item $KeyGen(\lambda)\rightarrow{K}$: Given the security parameter $\lambda$, the data owner generates the random secret key $K\in\{0,1\}^{\lambda}$ for PRFs $F_1$, $F_2$.
  
  \item $FreqSetup(\Gamma,r)\rightarrow{\mathcal{C}}$: The data owner selects a random integer $r\in\mathbb{N}$. Then it takes $r$ and the dictionary $\Gamma$ as input. For each keyword $w\in\mathcal{W}$, the data owner computes the ciphertext $ecnt_w$ for the plaintext $cnt_w$ composing the set $\mathcal{C}$, then it sends $\mathcal{C}$ to the cloud server. 
  
  \item $TokenGen(\mathcal{Q},K,r)\rightarrow{\mathcal{T}}$: The search user inputs a query $\mathcal{Q}$, the secret key $K$ and the random integer $r$. It generates a frequency token list $\mathcal{T}$, then sends $\mathcal{T}$ to the cloud server.
  
  \item $FreqFind(\mathcal{T},\mathcal{C})\rightarrow{\Delta}$: The cloud server inputs the token list $\mathcal{T}$ and set $\mathcal{C}$. It matches each token in $\mathcal{T}$ with the elements in $\mathcal{C}$ to get a matching result set $\Delta$. 
  
  \item $Compare(\Delta,K,r)\rightarrow(w_s,cnt_{w_s})$:  Given the set $\Delta$, the secret key $K$ and the random integer $r$, the search user outputs the least frequent keyword in the conjunction with its update frequency.
 
\end{itemize}

In the stage $FreqSetup$, to reduce frequent interaction between the search user and the data owner, for each $w\in\mathcal{W}$, the client computes the ciphertext as
\begin{eqnarray*}
ecnt_w=F_1(K,r||w)+F_2(K,w||r)+cnt_w,
\end{eqnarray*}
where $F_1$ generates $2\lambda$-bits output, $F_2$ generates $\lambda$-bits output, and operation $+$ is a bitwise addition operation, satisfying the value of $(F_2(K,w||r)+cnt_w)$ does not exceed $2^{\lambda}$, which is convenient to find $ecnt_w$ later. Therefore, when retrieving the encrypted value $ecnt_w$ corresponding to $w$, it is obtained that only the first $\lambda$-bits need to be matched according to $F_1(K,r||w)$. In addition, a fresh random integer $r$ is generated in each update and owned by the client, and $ecnt_w$ for each keyword is updated, ensuring the privacy of the updated keywords.

When a client requests to search the least frequent keyword in $\mathcal{Q}$, it runs $TokenGen$ to generate the frequency tokens. It computes $ftoken_w=F_1(K,r||w)$ for each keyword $w\in\mathcal{Q}$, and then sends $\mathcal{T}=\{ftoken_{w_1},\cdots,ftoken_{w_n}\}$ to the cloud server.

In response to the client, the cloud server executes the algorithm $FreqFind$. Note that the server uses $ftoken_w$ to match and obtain the corresponding $ecnt_w$ in the same range according to the first $\lambda$-bits. Then the server sends the set $\Delta=\{(ftoken_{w_i},ecnt_{w_i})\}_{i\in[n]}$ to the client.

The client runs the algorithm $Compare$ to find the keyword with the least update frequency. Specifically, it obtains $cnt_w$ by decrypting $ecnt_w, w\in\mathcal{Q}$ 
\begin{eqnarray*}
cnt_w=ecnt_w-F_1(K,r||w)-F_2(K,w||r),
\end{eqnarray*}
where the secret key $K~\text{and}~r$ are received previously from the data owner. Finally, the client compares the frequencies of each desired keyword in list $\mathcal{Q}$, and returns the least frequent keyword $w_s$ and its update frequency $cnt_{w_s}$.

\subsection{Setup Phase}\label{s_4_2}
\noindent \textbf{Client:} The \textsc{Setup} algorithm (Algorithm \ref{alg:1}) inputs the security parameter $\lambda$ and returns the $sk, st, param, \mathsf{EDB}$. The client initializes two empty maps $Cnt$ and $TSet$, where $Cnt$ is the state parameter $\sigma$ mentioned in Section \ref{s_3_1} storing update frequencies of keywords, and $TSet$ stores address-value entries. Then it generates a $m$-bits Bloom filter $BF$ with $k$ hash functions $\{H_j\}_{1\leq{j}\leq{k}}$, which is set up for dynamic cross-tags $xtag$ introduced later. The client randomly chooses keys for PRFs and  $msk$. Then it executes $\rm{SHVE.Setup}$ to encrypt $BF$. $TSet$ and $xtagBF$ are components of the encrypted database $\mathsf{EDB}$ and sent to the server.

\begin{algorithm}
    \algsetup{linenosize=\footnotesize} 
    \footnotesize
    \caption{Setup}
    \label{alg:1}
    \begin{algorithmic}[1]
        \REQUIRE a security parameter $\lambda$
        \ENSURE $sk, st, param, \mathsf{EDB}$
        \\$\textbf{\underline{\emph{ Client: }}}$
        \STATE Initialize $Cnt$, $TSet$ to empty maps
        \STATE Initialize a Bloom filter $BF$ with element 0
        \STATE Select $k$ hash functions $\{H_i\}_{1\leq{i}\leq{k}}$ for $BF$
        \STATE Select $K_T$ for PRF $F$
        \STATE Select $K_X, K_Y, K_Z$ for PRF $F_p$
        \STATE Execute LFKA.$KeyGen(\lambda)$ to get key for PRF $F_1,F_2$
        
        \STATE Execute $\rm{SHVE.Setup}(1^{\lambda})$ to get secret key $msk$
        \STATE Compute $xtagBF={\rm{SHVE.Enc}}(msk,\mu=``True",BF)$
        \STATE Set $sk=(msk, K, K_T, K_X, K_Y, K_Z),~st=Cnt,~param=\{H_i\}_{1\leq{i}\leq{k}}$
        \STATE Set $\mathsf{EDB(1)} = TSet$,  $\mathsf{EDB(2)} = xtagBF$
        \STATE Send $\mathsf{EDB}=(\mathsf{EDB(1)}, \mathsf{EDB(2)})$ to the server
    \end{algorithmic}
\end{algorithm}

\subsection{Upate Phase}\label{s_4_3}
\noindent \textbf{Client:} In the \textsc{Update} algorithm (Algorithm \ref{alg:2}), the client takes as input the key set $sk$, the state parameter $st$, and the set $L$ consisting of $(op, w, id)$. The server takes database $\mathsf{EDB}$ as input. The algorithm updates a batch of documents in each update. For each triple $(op,id,w)$ in $L$, the corresponding entry $(addr, val,\alpha)$ and dynamic cross-tag $xtag$ are generated adopting pseudorandom functions and exponentiations, then the counter $Cnt[w] = Cnt[w] + 1 $. The entry $(addr, val,\alpha)$ will be stored in the form of address-value pair in the $TSet$. The address $addr$ is generated by a pseudorandom function $F$ for $w$ and the current update frequency $cnt_w$. The corresponding stored contents of the address $addr$ are $val$ and $\alpha$, where $val$ is the encrypted $(id, op)$ pair related to $w$ and $\alpha$ is the dynamic blinding factor that enables the server to calculate the $xtag$. Besides, the client calculates the $k$ positions related to $xtag$ in the Bloom filter $BF$, and sets the corresponding elements to 1. Then it executes SHVE.Enc algorithm to encrypt $BF$, and gets the ciphertext $xtagBF$, and send the entries mentioned above to the server.

Note that in each update, the client is required to select a new random number $r'\in\mathbb{N}$ (to distinguish the previous $r$ from the fresh $r$, we use $r'$ to represent the latter) and re-encrypt the update frequencies of all keywords via LFKA.$FreqSetup$. Then it sends the fresh set $\mathcal{C}$ to the cloud server. According to the changed entry, this approach prevents the cloud server from deducing which keyword is in the updated file. Besides, it avoids frequent interactions between the data owner and the search user. Even though the cloud server might notice that the number of elements in set $\mathcal{C}$ changes after each update, it cannot infer the plaintext of the underlying keywords without knowing the secret key $K$. 

\begin{algorithm}
    \algsetup{linenosize=\footnotesize} 
    \footnotesize
    \caption{Update}
    \label{alg:2}
    \begin{algorithmic}[1]
        \REQUIRE $sk, st, L; \mathsf{EDB}$
        \ENSURE $st, \mathsf{EDB}$ 
        \\$\textbf{\underline{\emph{ Client: }}}$
        \STATE Initialize a dictionary $\Gamma$
        \STATE Initialize $\mathcal{C},addrList,valLsit,\alpha{List}$ to empty sets
        \STATE Get $sk=(msk, K, K_T, K_X, K_Y, K_Z)$, $st=Cnt$
        
            \FOR{$i=1:L.size$}
                \STATE $(op,id,w)=L[i]$
                \IF{$Cnt[w]=NULL$}
                    \STATE Set $Cnt[w]=0$
                \ENDIF
                \STATE Set $Cnt[w]=Cnt[w]+1$
                \STATE Set $\Gamma[w]=Cnt[w]$
                \STATE Compute $addr=F(K_T,w||Cnt[w]||0)$
                \STATE Set $addrList=addrList\cup\{addr\}$
                \STATE Compute $val=(id||op)\oplus{F(K_T,w||Cnt[w]||1)}$
                \STATE Set $valList=valList\cup\{val\}$
                \STATE Compute $\alpha=F_p(K_Y,id||op)\cdot{F_p(K_Z,w||Cnt[w])^{-1}}$
                \STATE Set $\alpha{List}=\alpha{List}\cup\{\alpha\}$
                \STATE Compute $xtag=g^{F_p(K_X,w)\cdot{F_p(K_Y,id||op)}}$
                    \FOR{$i=1:k$}
                        \STATE Compute $hind(op,id,w)=H_i(xtag)$ 
                        \STATE Set $BF[hind(op,id,w)]=1$
                    \ENDFOR
            \ENDFOR
        \STATE Select a random number $r\in\mathbb{N}$
        \STATE Set $\mathcal{C}={\rm{LFKA}}.FreqSetup(\Gamma,r)$
        
        \STATE Compute $xtagBF={\rm{SHVE.Enc}}(msk,\mu=``True",BF)$
        \STATE Send $(\mathcal{C},addrList,valLsit,\alpha{List},xtagBF)$  to the server
        \\$\textbf{\underline{\emph{ Server: }}}$
        \STATE Get $TSet=\mathsf{EDB(1)}$
            \FOR{$i=1:addrList.size$}
                \STATE Set $TSet[addrList[i]]=(valLsit[i],\alpha{List}[i])$
            \ENDFOR
        \STATE Set $\mathsf{EDB}=(TSet,xtagBF)$
    \end{algorithmic}
\end{algorithm}

\noindent \textbf{Server.} After receiving the encrypted entries, the server updates dictionary $TSet$ and encrypted Bloom filter $xtagBF$ accordingly. 

We introduce the update counter $cnt$ to ensure forward privacy. $cnt_w$ is incremented when an updated file contains the keyword $w$. Similarly, the random integer $r$ changes in each update. Because $cnt_w$ and $r$ are not the same as those in the latest search query, it ensures there is no relation between the new update and the previous search. Moreover, since the server learns nothing about the secret key $K$ and $K_T$, an update operation hides the underlying operation $op$, the identifier $id$, and the keyword $w$.

\subsection{Search Phase}\label{s_4_4}
Let $q=(w_1\wedge\cdots\wedge{w_n})$ be a conjunctive search query issued by the client, then the \textsc{SEARCH} algorithm (Algorithm \ref{alg:3}) involves three rounds of interaction between the client and the server. The keyword with the least frequency in the conjunction is denoted as the $s$-term $w_s$, and other keywords $w_i(2\leq{i}\leq{n})$ as the $x$-term.

Round-1 allows the client to execute the LFKA protocol with the server. Then it obtain the $s$-term $w_s$ and its frequency $cnt_{w_s}$. 

In Round-2, the client reverses the positions of $w_s$ and $w_1$, resulting in $w_s$ becoming $w_1$. Subsequently, it generates all relevant addresses $saddr$s in the $TSet$ involving $w_1$ from the first to the $cnt_{w_1}^{th}$ update. Meanwhile, it generates an additional set of cross-tokens $\{xtoken_{i,j}\}_{i\in[2,n],j\in[cnt_{w_1}]}$ related to the $s$-term $w_1$ and the $x$-term for each $id_{j}$, where $id_j$ is the identifier of the file involving the $s$-term. After receiving the list of $saddr$, the server retrieves the corresponding $TSet$ entries and gets the encrypted values $sval$ of $(id,op)$ pairs and the dynamic blinding factors $\alpha$. The server computes the relevant cross-tag $xtag_{i,j}$ involving $w_i$ and $id_j$ via exponential operations of cross-tokens and blinding factors. Then it calculates the positions of $BF$ which are set to 1 for storing $xtag_{i,j}$, and sends them to the client. 

Round-3 allows the client to construct a Bloom filter $vBF_j$ for $id_j$ based on the received positions. That is, we assume that each triple ${(op,id_j,w_i)}(2\leq{i}\leq{n})$ already exists in the database, where $op$ is the operation executed to the $id_j$ involving the $s$-term. Then the client executes SHVE.KeyGen to encrypt $vBF_j$, and sends the ciphertext $sBF_j$ to the server. The server performs membership test by executing ${\rm{SHVE.Query}}(sBF_j, xtagBF)$ to measure whether $id_j$ involves all keywords in the conjunction. If it returns $ ``True"$, all $n-1$ $x$-terms are also in the file $id_j$, and the server sends the relevant $sval$s to the client. 

Finally, the client decrypts $sval$s to recover $\mathsf{DB}(w_1)$ along with corresponding operations. Then it includes the identifiers with add operation in the result set $IdRList$, and discards those with delete operation from $IdRList$. 

\begin{algorithm*}
    \algsetup{linenosize=\footnotesize} 
    \footnotesize
    \caption{Search}
    \label{alg:3}
    \begin{multicols}{2}
    \begin{algorithmic}[1]
        \REQUIRE $param, sk, r, \mathcal{Q}=\{w_1,\cdots,{w_n}\}; \mathsf{EDB}$
        \ENSURE $IdRList$
        \\$\textbf{\underline{\emph{ Client: }}}$
        \STATE Initialize $\mathcal{T}$ to an empty set
        \STATE Get $sk=(msk, K, K_T, K_X, K_Y, K_Z)$
        \STATE Set $\mathcal{T}={\rm{LFKA}}.TokenGen(\mathcal{Q},K,r)$
        \STATE Send $\mathcal{T}$ to the server
        \\$\textbf{\underline{\emph{ Server: }}}$
        
        \STATE Initialize $\Delta$ to empty list
        \STATE Set $\Delta={\rm{LFKA}}.FreqFind(\mathcal{T},\mathcal{C})$ 
       
        \STATE Send set $\Delta$ to the client
        \\$\textbf{\underline{\emph{ Client: }}}$
        
        \STATE Set $(w_s,cnt_{w_s})={\rm{LFKA}}.Compare(\Delta,K,r)$
        \STATE Replace $w_s$ to the front of the list $\mathcal{Q}$ as $w_1$

        \STATE Initialize $saddrList$ and $xtokenList[cnt_{w_1}]$to empty lists
            \FOR{$j=1:cnt_{w_1}$}
               
                \STATE Compute $saddr_j=F_p(K_T,w_1||j||0)$
                \STATE Set $addrList=saddrList\cup\{saddr_j\}$
                
                \FOR{$i=2:n$}
                    \STATE Compute $xtoken_{i,j}=g^{F_p(K_Z,w_1||j)\cdot{F_p(K_X,w_i)}}$
                    \STATE $xtokenLsit[j]=xtokenList[j]\cup\{xtoken_{i,j}\}$
                \ENDFOR
                \STATE Randomly shuffle the tuple-entries of $xtokenList[j]$
            \ENDFOR
        \STATE Send $saddrList,xtokenList$ to the server
        \\$\textbf{\underline{\emph{ Server: }}}$
        \STATE Get $(TSet,xtagBF)=\mathsf{EDB}$
        \STATE Initialize $vBFind[xtokenList.size]$ to empty lists
            \FOR{$j=1:saddrList.size$}
                \STATE Set $(sval_j,\alpha_j)=TSet[saddrList[j]]$
                
                \FOR{$i=2:n$}
                    \STATE Set $xtoken_{i,j}=xtokenList[j][i-1]$
                    \STATE Compute $xtag_{i,j}=xtoken_{i,j}^{\alpha_j}$
                    \FOR{$t=1:k$}
                        \STATE Compute $hind=H_t(xtag_{i,j})$
                        \STATE Set $vBFind[j]=vBFind[j]\cup{\{hind\}}$
                    \ENDFOR
                \ENDFOR
            \ENDFOR
        \STATE Send $vBFind[1],…,vBFind[j]$ to the client
        \\$\textbf{\underline{\emph{ Client: }}}$
        \STATE Initialize $cnt_{w_1}$ Bloom filters $vBF$ with element $\ast$
        \STATE Initialize $sBFList$ to an empty list
            \FOR{$j=1:cnt_{w_1}$}
                \FOR{$c=1:vBFind[j].size$}
                    \STATE Set $vBF[j][vBFind[c]]=1$
                \ENDFOR
                \STATE Compute $sBF_j={\rm{SHVE.KeyGen}}(msk,vBF[j])$
                \STATE Set $sBFList= sBFList\cup\{sBF_j\}$
            \ENDFOR
        \STATE Send $sBFList$ to the server
        \\$\textbf{\underline{\emph{ Server: }}}$
        \STATE Initialize $sRList$ to an empty list
       
            \FOR{$j=1:sBFList.size$}
                \STATE Compute $res_j={\rm{SHVE.Query}}(sBFList[j],xtagBF)$
                \IF{$res_j=``True"$}
                    \STATE Add $(j,sval_j)$ into $sRList$
                \ENDIF
            \ENDFOR
            \STATE Send $sRList$ to the client
        \\$\textbf{\underline{\emph{ Client:Final}}}$
        \STATE Initialize $IdRList$ to an empty list
            \FOR{$i=1:sRList.size$}
                \STATE Set $(j,sval_j)=sRList[i]$
                \STATE Compute  $(id_j,op_j)=sval_j\oplus{F(K_T,w_1||j||1)}$
                \IF{$op_j==add$}
                    \STATE Set $IdRList=IdRList\cup\{id_j\}$
                \ELSIF{$op_j==del$}
                    \STATE Set $IdRList=IdRList\backslash\{id_j\}$
                \ENDIF
            \ENDFOR
    \end{algorithmic}
	\end{multicols}
\end{algorithm*}

Recall that the data owner encrypts and uploads the frequency of updates to the server in the latest update. Therefore, the search user can ask the server for the least frequent keyword, which eliminates frequent interactions between clients in searches, and handing over the corresponding task to the server with unlimited computation and storage capacities. 

We now elaborate the correctness of oblivious conjunctive search. Over any number of update operations, the dynamic cross-tag can evaluate the presence or absence of any identifier-keyword pair $(id,w)$ in a dynamic dataset. More specifically, for an update triple $(op,id_j,w_i)$, there exists the cross-tag
\begin{eqnarray*}
xtag_{op,i,j}=g^{F_p(K_X,w_i)\cdot{F_p(K_Y,id_j||op)}},
\end{eqnarray*}
where $op\in\{add,del\}$. The $xtag$ value is stored in $BF$. We use a Bloom filter containing $\{xtag_{i,j}\}_{i\in[2,n]}$ to assume that file $id_j$ involves all $n-1$ $x$-terms. By performing SHVE protocol, the server checks the membership of $id_j$ in ciphertext without the KPRP leakage on which $w_i$ are in $id_j$.

To make the server obliviously compute the cross-tag, the client as a search user generates and sends an additional set of cross-tokens $\{xtoken_{i,j}\}_{i\in[2,n],j\in[cnt_{w_1}]}$ to the server which involve $s$-term $w_1$ and $n-1$ $x$-terms. Then, we can obtain
\begin{gather*}
xtoken_{i,j}=g^{F_p(K_X,w_i)\cdot{F_p(K_Z,w_1||j)}}.
\end{gather*}

Since the $TSet$ address corresponding to the
$j^{th}$ update operation involving $w_1$ stores an additional pre-computed blinding factor $\alpha$, the server can calculate $xtag_{op,i,j}$ using blinded exponentiations as
\begin{gather*}
xtag_{op,i,j}=g^{F_p(K_X,w_i)\cdot{F_p(K_Y,id_j||op)}}=(xtoken_{i,j})^{\alpha}.
\end{gather*}

The aforementioned computation process shows that without learning underlying update triple $(op,id,w)$, the server can obliviously compute the relevant cross-tag. Next, if all elements in $\{xtag_{op,i,j}\}_{i\in[2,n]}$ exist in $xtagBF$, the file $id_j$ involving $w_1$ also contains other $n-1$ keywords. 

\section{Security Analysis}\label{s_5}
In this section, we evaluate the security of our scheme. Firstly, we formally describe the leakage functions for ESP-CKS. Subsequently, we demonstrate its forward, backward, and keyword pair result privacy.

\subsection{Leakage Functions}\label{s_5_1}
We aim to guarantee that the DSSE scheme reveals as little information as possible, ensuring that it achieves a higher level of security. The leakage functions capture leakage. Similar to \cite{patranabis2021forward}, we formally define the leakage functions as
\begin{eqnarray*}
\mathcal{L}=(\mathcal{L}_{Setup},\mathcal{L}_{Update},\mathcal{L}_{Search}),
\end{eqnarray*}
where $\mathcal{L}_{Setup}= \perp$, $\mathcal{L}_{Update}(op,id,w)=\perp$, $\mathcal{L}_{Search(q)}=(\mathsf{TimeDB}(q),\mathsf{Upd}(q))$.

Let $\mathcal{O}$ be a list containing the following entries:
\begin{enumerate}
    \item $(t, w)$: search query for keyword $w$ at timestamp $t$;
    \item $(t, op, id, w)$: update query for $(op, id, w)$ at timestamp $t$, where $op\in\{add,del\}$.
\end{enumerate}

For any conjunctive query $q=(w_1\wedge\cdots\wedge{w_n})$, we define $\mathsf{TimeDB}(q)$ as a function that returns file identifiers and insertion timestamps. Files corresponding to the identifiers contain all the keywords involved in the query $q$ and have not yet been deleted from the database. Suppose the keyword $w_1$ is the $s$-term, namely aforementioned least frequent term $w_s$. Formally, it is expressed as
\begin{eqnarray*}
\mathsf{TimeDB}(q)=\{(\{t_i\}_{i\in[n]}, id)\,|\,(t_i,add,id,w_i)\in\mathcal{O} \\ and\quad\forall t^{\prime}: (t^{\prime},del,id,w_i)\notin\mathcal{O}\}.
\end{eqnarray*}

We define $\mathsf{Upd}(q)$ as a function that returns the timestamps of all update operations on the $s$-term $w_1$. Formally, it is expressed as
\begin{eqnarray*}
\mathsf{Upd}(q)=\{t\,|\,\exists(op,id):(t,op,id,w_1)\in\mathcal{O} \}.
\end{eqnarray*}

For simplicity, we assume that no Bloom filter false positives occur in our protocol. We have the following theorem for the security of our scheme.

\begin{theorem}\label{thm1}
Our ESP-CKS is adaptively-secure with respect to a leakage function $\mathcal{L}$ defined as before, assuming that $F, F_1, F_2, F_p$ are secure pseudorandom functions, the decisional Diffie-Hellman assumption holds over $\mathbb{G}$, and SHVE is a selectively simulation-secure protocol.
\end{theorem} 

\begin{proof}
The detailed proof is given in Appendix.
\end{proof}

\subsection{Forward Privacy}\label{s_5_2}
According to the forward privacy definition introduced in \cite{bost2017forward}, the dynamic conjunctive SSE is adaptively forward privacy iff the update leakage function $\mathcal{L}_{Update}$ can be written as
\begin{eqnarray*}
\mathcal{L}_{Update}(op,id,w)=\mathcal{L}^{\prime}(op,id),
\end{eqnarray*}
where $\mathcal{L}^{\prime}$ is a stateless function. Note that forward privacy guarantees that the underlying keyword $w$ is hidden during update phase. Whereas, in our scheme $\mathcal{L}_{Update}(op,id,w)=\perp$. That is, an update operation hides the underlying keyword $w$, along with the identifier $id$ and the operation $op$. A natural corollary of Theorem \ref{thm1} is as follows.

\begin{corollary}[Forward Privacy of ESP-CKS]
Our scheme is adaptively forward private if $F, F_1, F_2, F_p$ are secure pseudorandom functions, the decisional Diffie-Hellman assumption holds over $\mathbb{G}$, and SHVE is a selectively simulation-secure protocol.
\end{corollary}

\subsection{Backward Privacy}\label{s_5_3}
Backward privacy is formally defined in  \cite{bost2017forward} and is ordered from the most to least secure as Type-I, Type-II, Type-III. According to this, a DSSE scheme supporting single keyword search is adaptively backward privacy iff the update leakage function $\mathcal{L}_{Update}$ and the search leakage function $\mathcal{L}_{Search(w)}$ can be written as
\begin{gather*}
\mathcal{L}_{Update}(op,id,w)=\mathcal{L}^{\prime\prime}(op,id), \text{and}   \\
\mathcal{L}_{Search(w)}=\mathcal{L}^{\prime\prime\prime}(\mathsf{TimeDB}(w),\mathsf{Upd}(w)),
\end{gather*}
where $\mathcal{L}^{\prime\prime}$ and $\mathcal{L}^{\prime\prime\prime}$ are stateless functions.

For a DSSE scheme supporting conjunctive keyword search, a natural generalization of the aforementioned leakage function is defined as
\begin{gather*}
\mathcal{L}_{Update}(op,id,w)=\perp, \text{and}   \\
\mathcal{L}_{Search(q)}=(\mathsf{TimeDB}(q),\mathsf{Upd}(q)),
\end{gather*}
where $q$ is a conjunctive search query. The leakage meets the Type-II backward privacy. A natural corollary of Theorem \ref{thm1} is as follows.

\begin{corollary}[Backward Privacy of ESP-CKS] 
Our scheme is Type-II adaptively backward private if $F, F_1, F_2, F_p$ are secure pseudorandom functions, the decisional Diffie-Hellman assumption holds over $\mathbb{G}$, and SHVE is a selectively simulation-secure protocol.
\end{corollary}

\subsection{Keyword Pair Result Privacy}\label{s_5_4}
We compare ESP-CKS with ODXT scheme \cite{patranabis2021forward} to prove that our scheme guarantees keyword pair result privacy; that is, it prevents the KPRP leakage. We consider the following example to illustrate the privacy improvement of the proposed scheme. 

Assuming that the database in the server stores five encrypted files, whose plaintext of indexes are $\{id_i\}_{1\leq{i}\leq{5}}$. The set of keywords $\mathcal{W}$ contains six keywords expressed as $\{w_i\}_{1\leq{i}\leq{6}}$. Each file involves different keywords, as shown in Table \ref{kprp_db}.

\begin{table}[htbp]
	\centering
	\begin{small}
		\caption{\label{kprp_db} Document Identifier’s and Keywords Contained.}
		\centering
		\begin{tabular}{|c|c|} 	  
			\hline
			\textbf{identifier} & \textbf{keywords} \\
			\hline
			\hline
			1 & $w_1,w_2,w_3,w_4,w_6$\\
			2 & $w_2,w_3,w_4,w_5$\\
			3 & $w_1,w_2,w_4,w_5,w_6$\\
			4 & $w_2,w_3,w_6$\\
			5 & $w_1,w_3,w_4$\\
			\hline
		\end{tabular} 
	\end{small}
\end{table}

Suppose that the client issues the search query $q=(w_1\wedge{w_2}\wedge{w_3})$. Note that keyword $w_1$ is the least frequent term. The set of documents involving $w_1$ is $\mathsf{DB}(w_1)=\{id_1,id_3,id_5\}$. Then the server determines whether the files contain $w_2$ and $w_3$. In ESP-CKS, the server only learns the result $\mathsf{DB}(q)=\{id_1\}$. However, in ODXT, the server can learn the partial query result that $\mathsf{DB}(w_1)\cap\mathsf{DB}(w_2)=\{id_1,id_3\}$ and $\mathsf{DB}(w_1)\cap\mathsf{DB}(w_3)=\{id_1,id_5\}$. Finally, the server sends the encrypted data of the intersection of these two sets (namely $\{id_1\}$) to the client. Although the partial query result called KPRP leakage is under the ciphertext, \cite{zhang2016all} shows that file-injection attacks can reveal all keywords of a conjunctive query with 100\% accuracy via utilizing KPRP leakage. Our ESP-CKS scheme eliminates the aforementioned KPRP leakage and improves privacy.

\section{Performance Analysis}\label{s_6}
In this section, we analyze the performance of the ESP-CKS scheme and compare it with existing schemes. We evaluate the performance theoretically and experimentally.

\subsection{Theoretical Analysis}\label{s_6_1}
We now discuss the practical storage, computation, communication costs of our scheme and compare them with that in conjunctive keyword search schemes ODXT \cite{patranabis2021forward} and SE-EPOM \cite{liu2020privacy}. Notably, there are two types of servers in SE-EPOM. Here we only focus on the performance of the Cloud Platform (CP). We give a list of notations for comparative analysis in Table \ref{tabno}.

\begin{table*}[htbp]
	\renewcommand{\multirowsetup}{\centering}
	\renewcommand\arraystretch{1.3}
	\begin{small}
		\caption{\label{tabcc} Comparison of Conjunctive Keyword Search Schemes.}
		\centering
		\begin{tabular}{|p{2cm}<{\centering}|p{2cm}<{\centering}|p{1.5cm}<{\centering}|p{0.25\columnwidth}<{\centering}|p{3cm}<{\centering}|p{4.8cm}<{\centering}|}
			\hline
			\multicolumn{3}{|c|}{}                                                                                                                                      & SE-EPOM \cite{liu2020privacy}                                           & \multicolumn{1}{c|}{ODXT \cite{patranabis2021forward}}                 & \multicolumn{1}{c|}{OURS}                                                                                                                     \\ \hline	\hline
			\multicolumn{1}{|c|}{\textbf{storage}}                       & \multicolumn{2}{c|}{\textbf{storage size(S)}}                                                & $M\mathsf{Z_{\lambda^2}}$                         & \multicolumn{1}{c|}{$~~~~~~N\lambda+m~~~~~~$}         & \multicolumn{1}{c|}{$N\lambda+m\lambda+|\mathcal{W}|\lambda$}                                                                                 \\ \hline 	\hline
			\multicolumn{1}{|c|}{\multirow{7}{*}{\textbf{computation}}}  & \multicolumn{1}{c|}{\multirow{2}{*}{\textbf{update cost}}}     & common                      & \multirow{2}{*}{$\mathsf{m}(e+T_\mathsf{MUL^2})$} & \multicolumn{2}{c|}{$\mathsf{N}(T_\mathsf{PRF}+T_\mathsf{MUL}+T_\mathsf{XOR}+e_{pre}+kT_\mathsf{hash})$}                                                                                  \\ \cline{3-3} \cline{5-6} 
			\multicolumn{1}{|c|}{}                                       & \multicolumn{1}{c|}{}                                          & additional                  &                                                   & \multicolumn{1}{c|}{N/A}                  & \multicolumn{1}{c|}{$(\mathsf{u}+m)T_\mathsf{PRF}+\mathsf{u}T_\mathsf{ADD}$}                                                 \\ \cline{2-6} 
			\multicolumn{1}{|c|}{}                                       & \multicolumn{1}{c|}{\multirow{2}{*}{\textbf{search cost (S)}}} & common                      & $M(\theta+1)\cdot$                           & \multicolumn{2}{c|}{$C(n-1)(e+kT_\mathsf{hash})$}                                                                                                                                         \\ \cline{3-3} \cline{5-6} 
			\multicolumn{1}{|c|}{}                                       & \multicolumn{1}{c|}{}                                          & additional                  & $(e+T_\mathsf{MUL^2})$                        & \multicolumn{1}{c|}{N/A}                  & \multicolumn{1}{c|}{$C(m'T_\mathsf{XOR}+T_\mathsf{Dec})$}                                                                                     \\ \cline{2-6} 
			\multicolumn{1}{|c|}{}                                       & \multicolumn{1}{c|}{\multirow{3}{*}{\textbf{search cost (C)}}} & common                      & \multirow{3}{*}{$M(e+T_\mathsf{MUL^2})$}          & \multicolumn{2}{c|}{$C((n+1)T_\mathsf{PRF}+(n-1)e_{pre}+T_\mathsf{XOR})$}                                                                                                                 \\ \cline{3-3} \cline{5-6} 
			\multicolumn{1}{|c|}{}                                       & \multicolumn{1}{c|}{}                                          & \multirow{2}{*}{additional} &                                                   & \multicolumn{1}{c|}{\multirow{2}{*}{N/A}} & $n(T_\mathsf{PRF}+T_\mathsf{ADD})+T_\mathsf{Comp}+$ \\
			\multicolumn{1}{|c|}{}                                       & \multicolumn{1}{c|}{}                                          &                             &                                                   & \multicolumn{1}{c|}{}                     & $C(m'T_\mathsf{PRF}+m'T_\mathsf{XOR}+T_\mathsf{Enc})$                                                                                         \\ \hline	\hline
			\multicolumn{1}{|c|}{\multirow{2}{*}{\textbf{comunication}}} & \multicolumn{1}{c|}{\multirow{2}{*}{\textbf{cost}}}            & common                      & \multirow{2}{*}{$MO(\lambda^2)$}                  & \multicolumn{2}{c|}{$C((n-1)\mathsf{G}+\lambda)+\mathsf{r}O(\lambda)$}                                                                                                                    \\ \cline{3-3} \cline{5-6} 
			\multicolumn{1}{|c|}{}                                       & \multicolumn{1}{c|}{}                                          & additional                  &                                                   & \multicolumn{1}{c|}{N/A}                  & \multicolumn{1}{c|}{$nO(\lambda)+CO(m'+2\lambda)$}                                                                                        \\ \hline
		\end{tabular}
	\end{small}
\end{table*}

\begin{table}[htbp] 
\centering
\begin{small}
\caption{\label{tabno} Notations for Comparative Analysis.}
\centering
\begin{tabular}{|c|l|} 
\hline  
\textbf{Notation} & \textbf{Meaning} \\
\hline
\hline
$C$ & update frequency of the $s$-term \\
\hline
$e$ & number of exponentiations \\
\hline
$e_{pre}$ & number of preprocessed exponentiations \\ 
\hline
$\mathsf{u}$ & number of keywords updated in a batch \\
\hline
$\mathsf{N}$ & number of triple $(op,id,w)$ updated in a batch \\
\hline
$\mathsf{m}$ & number of files updated in a batch \\
\hline
$\mathsf{r}$ & number of pairs $(id,op)$ in the search result \\
\hline
$\theta$ & the maximum number of keywords in set $\mathcal{W}$\\
\hline
$\mathsf{G}$ & size of an element from $\mathbb{G}$ \\
\hline
$\mathsf{Z_{\lambda^2}}$ & size of an element from $\mathbb{Z}_{\lambda^2}^{\ast}$ \\
\hline
$T_\mathsf{PRF}$ & time taken to compute a PRF \\
\hline
$T_\mathsf{XOR}$ & time taken to execute an XOR operation over $\lambda$ \\
\hline
$T_\mathsf{hash}$ & time taken to compute a hash of Bloom filter \\
\hline
$T_\mathsf{ADD}$ & time taken to perform an addition over $\lambda$ \\
\hline
$T_\mathsf{MUL}$ & time taken to perform a multiplication over $\lambda$ \\
\hline
$T_\mathsf{MUL^2}$ & time taken to perform a multiplication over $\lambda^2$ \\
\hline
$T_\mathsf{Comp}$ & time taken to compare $n$ values. \\
\hline
$T_\mathsf{Enc}$ & time taken to perform a symmetric encryption. \\
\hline
$T_\mathsf{Dec}$ & time taken to perform a symmetric decryption. \\
\hline  
\hline
\end{tabular} 
\end{small}
\end{table}

\noindent \textbf{Storage:} Focusing on the storage size, the server in ODXT stores the dictionary $TSet$ and the Bloom filter of $XSet$. For a triple $(op,id,w)$, a $O(\lambda)$-size entry of the form $(addr,\alpha,val)$ is added to $TSet$, and the size of $XSet$ is $m$. In comparison to ODXT, we adopt ciphertext of a Bloom filter with $m\lambda$-size instead of $XSet$. Moreover, the server stores the ciphertext of keyword update frequency of size $|\mathcal{W}|\lambda$. If update operations contain $N$ triples $(op,id,w)$, the overall storage size of our ESP-CKS scheme is $N\lambda+m\lambda+|\mathcal{W}|\lambda$. In SE-EPOM, server stores $M$ entries and each encrypted by Enc algorithm \cite{LiuXimeng2016} of $\mathsf{Z_{\lambda^2}}$-size. Hence, the storage cost of SE-EPOM is more than the proposed scheme. 

\noindent \textbf{Update Computation Overhead:} During the update phase, the computational cost is generated by the client. For a batch update containing multiple triples $(op,id,w)$, in our scheme, the time taken to encrypt keywords update frequencies is $\mathsf{u}T_\mathsf{PRF}+\mathsf{u}T_\mathsf{ADD}$. For the computation in $TSet$, the cost of generating $\mathsf{N}$ entries $(addr,val,\alpha)$ is $\mathsf{N}(T_\mathsf{PRF}+T_\mathsf{XOR}+T_\mathsf{MUL})$. 
For the generation of $xtag$s and $xtagBF$, the costs are $\mathsf{N}e_{pre}+\mathsf{N}kT_\mathsf{hash}$ and $mT_\mathsf{PRF}$, respectively. Note that the ODXT shares a lot of similarities with our scheme on the stage of $TSet$ computation and the Bloom filter generation involving $xtag$s, the cost is $\mathsf{N}(T_\mathsf{PRF}+T_\mathsf{MUL}+T_\mathsf{XOR}+e_{pre}+kT_\mathsf{hash})$.
Although forward and backward privacy cannot be guaranteed, the static SE-EPOM protocol can also perform adding files operation. In SE-EPOM, the client encrypts a decimal number of keywords in a file with the proposed public key encryption method, which involves exponentiations and multiplications, and the overall update cost is $\mathsf{m}(e+T_\mathsf{MUL^2})$.

\noindent \textbf{Search Computation Overhead:} We now investigate the search computation overhead of our ESP-CKS and contrast it with that of the ODXT and SE-EPOM. Both server-side and client-side should be considered during the search phase.

First, we focus on the computational cost of our scheme. The client spends $n(T_\mathsf{PRF}+T_\mathsf{ADD})+T_\mathsf{Comp}$ when it generates frequency tokens. In the subsequent interaction, the client computes $saddr$s and $xtoken$s costing $C(nT_\mathsf{PRF}+(n-1)e_{pre})$. 
The server spends $C(n-1)e$ to compute $xtag$s and $C(n-1)kT_\mathsf{hash}$ for membership test in Bloom filters. The client computes $\mathsf{sBF}$ through $\mathsf{SHVE.KeyGen}$, which costs $m'T_\mathsf{PRF}+m'T_\mathsf{XOR}+T_\mathsf{Enc}$. 
The server spends $C(m'T_\mathsf{XOR}+T_\mathsf{Dec})$ to perform $\mathsf{SHVE.Query}$. Finally, the client decrypts the search result costing $C(T_\mathsf{PRF}+T_\mathsf{XOR})$. The search computation overhead investigated above is presented in Table \ref{tabcc}.

In ODXT, the server spends $C(n-1)(e+kT_\mathsf{hash})$ to generate $xtag$s and calculate the location of elements in the Bloom filter. The client spends $C(nT_\mathsf{PRF}+(n-1)e_{pre})$ in $saddr$s and $xtoken$s generation, and $C(T_\mathsf{PRF}+T_\mathsf{XOR})$ in result decryption, respectively.

We now compare the computational cost on the server-side and client-side between our ESP-CKS scheme and the ODXT scheme. We define $O_{s}$ on the server-side as
\begin{eqnarray*}
O_{s}=\frac{m'T_\mathsf{XOR}+T_\mathsf{Dec}}{(n-1)(e+kT_\mathsf{hash})}.
\end{eqnarray*}

According to the relationship between execution time of different operations mentioned in \cite{lai2018result}, we conclude that additional cost on server side is only about $1.8\%$, while false positive rate is $10^{-6}$, $n=2$, $k\approx{\log_2(1/P_e)}\approx20$, and $m'=(n-1)k$.

Similarly, we define $O_{c}$ on the client side as
\begin{eqnarray*}
O_{c}=\frac{T_{freq}+T_{sBF}}{C((n+1)T_\mathsf{PRF}+(n-1)e_{pre}+T_\mathsf{XOR})},
\end{eqnarray*}
where $T_{freq}=n(T_\mathsf{PRF}+T_\mathsf{ADD})+T_\mathsf{Comp}$ and $T_{sBF}=C(m'T_\mathsf{PRF}+m'T_\mathsf{XOR}+T_\mathsf{Enc})$. Likewise, we can deduce that additional cost on server side is about $15\%$ when $C=10$, and it decreases as $n$ and $C$ increase. Therefore, our scheme achieves an enhanced security with minimum compromise in search efficiency.

In SE-EPOM, the client needs to generate a trapdoor and decrypt the results with a weak private key, which costs $M(e+T_\mathsf{MUL^2})$. On the server-side, the computation overhead of protocol SBD \cite{samanthula2013efficient}, SAD, and SMD is $\theta(e+T_\mathsf{MUL^2})$. Moreover, NOT-operation in ciphertext takes $e+T_\mathsf{MUL^2}$. Hence, the overall search computational cost is $M(\theta+1)(e+T_\mathsf{MUL^2})$ for the server. Note that the server computation overhead is linear with $C$ and $n$ in our scheme, but scales with $M$ and $\theta$ in SE-EPOM. It can be observed that the search in our scheme is more efficient because of $M\geq C$, $\theta\geq n$ and the probability of being equal is extremely low.

\noindent \textbf{Communication Overhead:} Focusing on communication cost in our scheme, the bandwidth of sending encrypted frequency values is $nO(\lambda)$. The interaction of transmit $saddr$s and $xtoken$s costs $C\lambda+C(n-1)\mathsf{G}$. The server sends the positions in $BF$ to the client with $CO(m')$ communication overhead. Then the client sends the output of $\mathsf{SHVE.KeyGen}$ with $C(O(m')+2\lambda)$ cost. Finally, the server costs $\mathsf{r}O(\lambda)$ to transmit the final result. Note that in ODXT, there is no communication about transmitting frequency values and ciphertext of Bloom filter, and the overall communication overhead is $C\lambda+C(n-1)\mathsf{G}+\mathsf{r}O(\lambda)$. In SE-EPOM, the client sends the search token with $O(\lambda^2)$, and the server sends encrypted result with $MO(\lambda^2)$. The communication overhead investigated above is presented in Table~\ref{tabcc}.

\subsection{Experimental Evaluations}\label{s_6_2}
We conduct the following experiments to evaluate the efficiency of our scheme in terms of setup, update, and conjunctive search time. Each experimental result is the average execution times over 10 runs.

\begin{figure}[!t]
	\centering
	\subfigure[$|CNT(w_2)|=2^{10}$]{\includegraphics[width=1.7in]{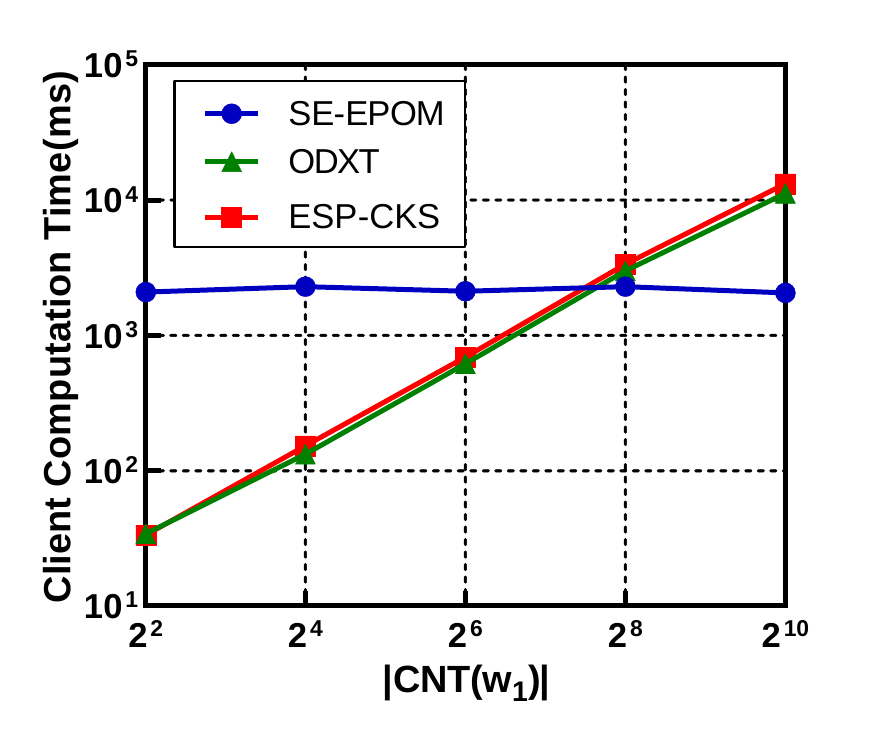}\label{Fig_2_1}}
	\hfil
	\subfigure[$|CNT(w_1)|=2^{2}$]{\includegraphics[width=1.7in]{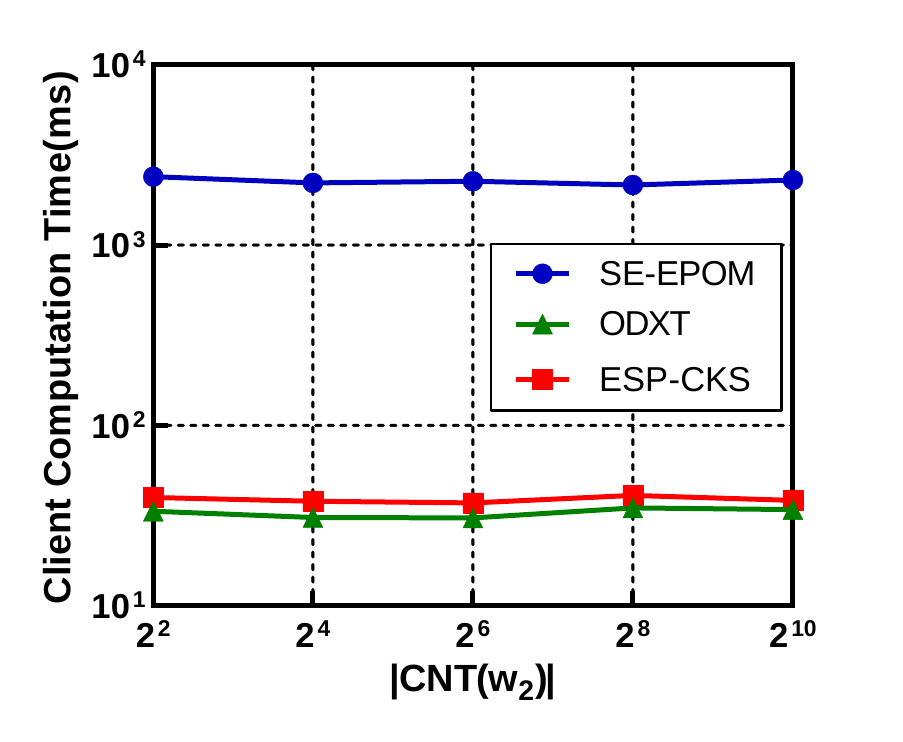}\label{Fig_2_2}}
	\vfil
	\subfigure[$|CNT(w_2)|=2^{10}$]{\includegraphics[width=1.7in]{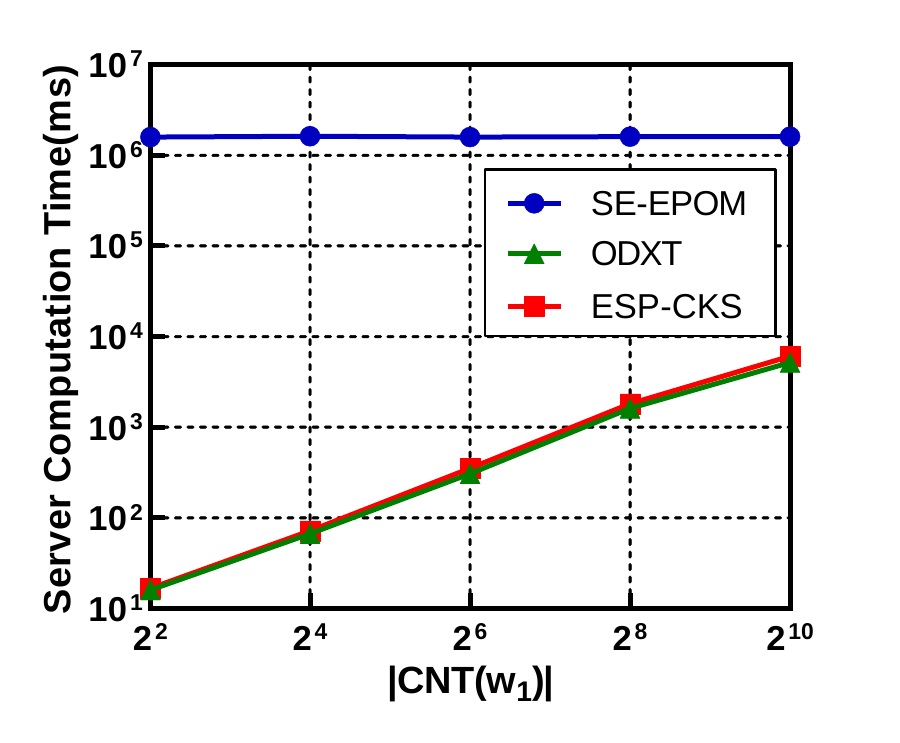}\label{Fig_2_3}}
	\hfil
	\subfigure[$|CNT(w_1)|=2^{2}$]{\includegraphics[width=1.7in]{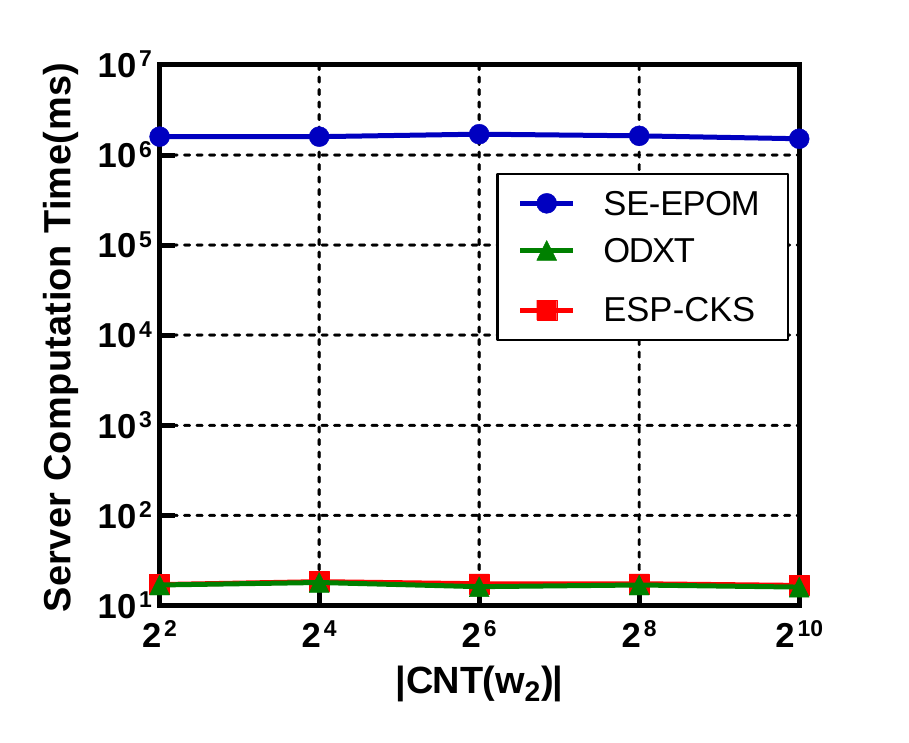}\label{Fig_2_4}}
	\hfil
	\caption{Two-conjunctive search query $q=(w_1\wedge{w_2})$ computation time over client and server.}
	\label{Fig_2}
\end{figure}

\begin{table}[htbp] 
	\caption{\label{overhead} Overhead Comparison.}
	\centering
	\begin{tabular}{c|c|c|c} 
		\hline  
		& SE-EPOM & ODXT & OURS \\
		\hline
		setup(ms) & 4006 & 416 & 468\\
		\hline  
		update-1000(ms) & 2.91 & 55.27 & 61.01\\
		\hline  
		update-10000(ms) & / & 42.34 &73.21 \\
		\hline
	\end{tabular} 
\end{table}

\noindent \textbf{Implementation details.} We implement our ESP-CKS with Python. The communication is simulated in a single-threaded environment to facilitate testing time. We conduct our experiments on a PC with a 1.30 GHz eight-core processor and 16GB RAM. The security parameter is $\lambda = 256$. We deploy the scalable Bloom filter from Paulo S{\'e}rgio Almeida \cite{almeida2007scalable}. The false positive rate of the Bloom filter is set to $10^{-6}$, which is negligible and enables the server to perform the membership test accurately. For comparison, we also implement conjunctive keyword search schemes ODXT \cite{patranabis2021forward} and SE-EPOM \cite{liu2020privacy}. Notation $\theta$ denotes the maximum number of keywords in set $\mathcal{W}$. The outputs of $F, F_p, F_2$ are 256 bits long, while that of $F_1$ is 512 bits long.

\noindent \textbf{Datasets:} 
We employ two different sorts of datasets: one with 1,200 files and the other with 12,500 files. Because SE-EPOM runs slow in a single thread on an extensive database without supporting deleting files, the dataset with 1,200 files is exploited to compare the three works. The dataset with 12,500 files is to compare our ESP-CKS with ODXT. We ensure that each document in datasets is inserted at least once. Update operations can be additions or deletions of batch files, 20\% of which are deletions. 

\begin{figure}[!t]
	\centering
	\subfigure[$|CNT(w_2)|=2^{10}$]{\includegraphics[width=1.7in]{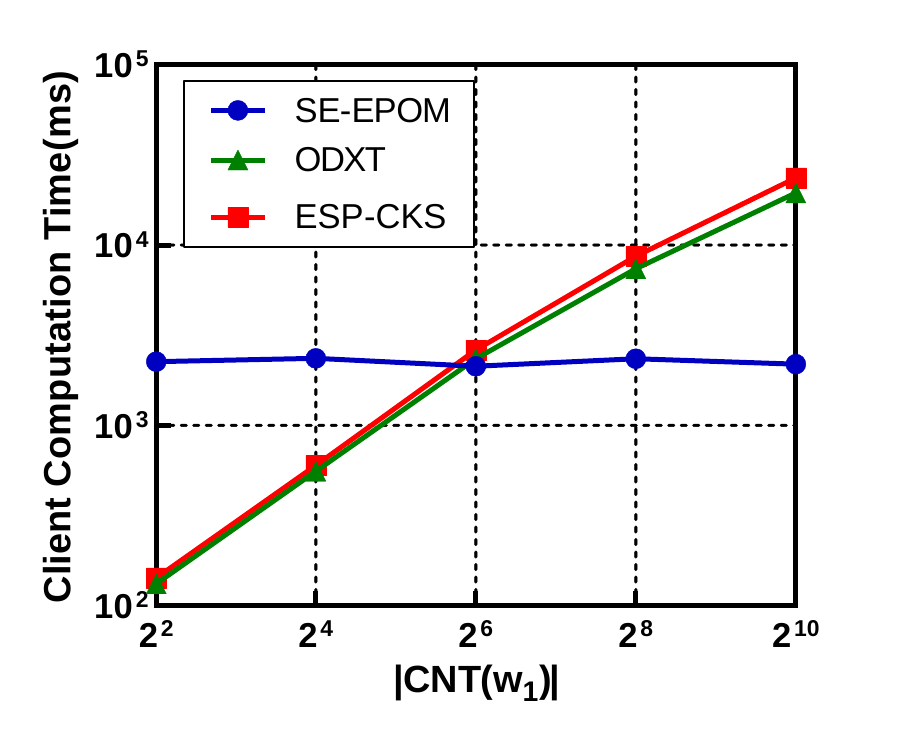}\label{Fig_3_1}}
	\hfil
	\subfigure[$|CNT(w_1)|=2^{2}$]{\includegraphics[width=1.7in]{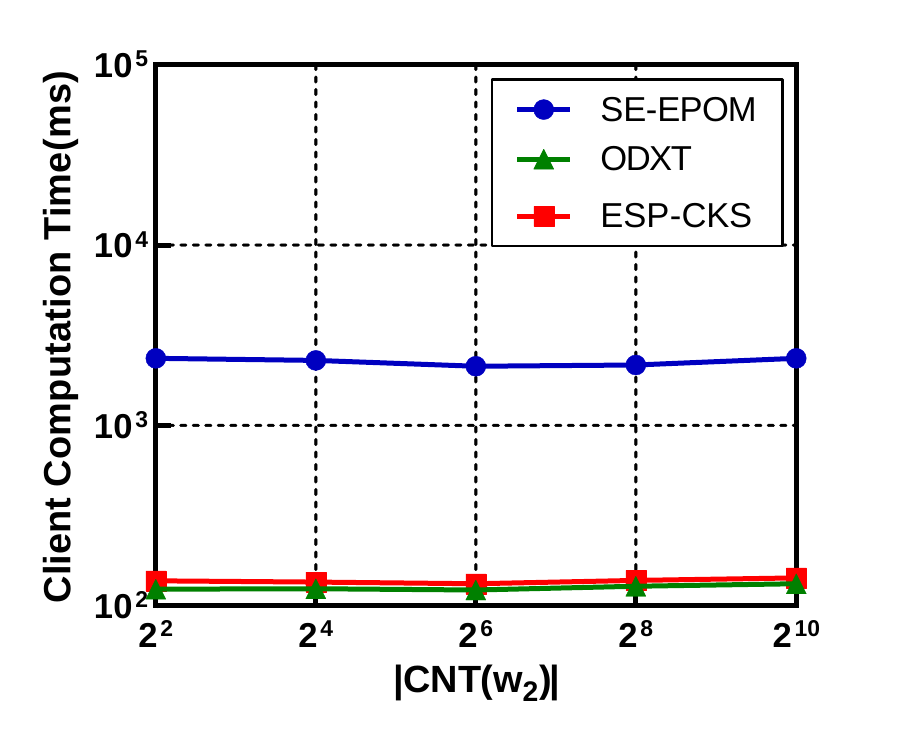}\label{Fig_3_2}}
	\vfil
	\subfigure[$|CNT(w_2)|=2^{10}$]{\includegraphics[width=1.7in]{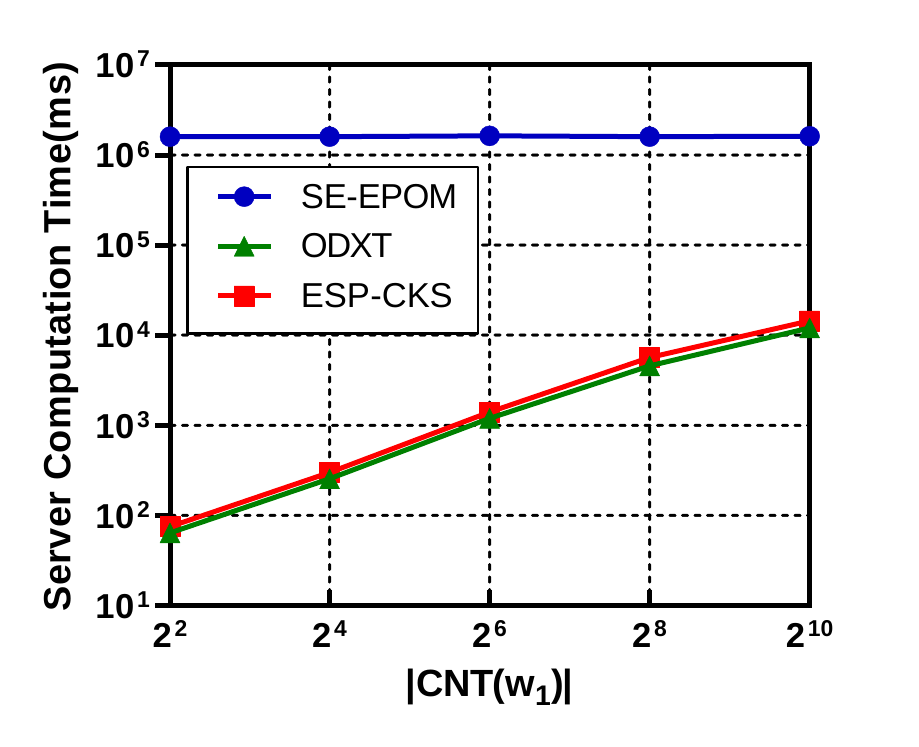}\label{Fig_3_3}}
	\hfil
	\subfigure[$|CNT(w_1)|=2^{2}$]{\includegraphics[width=1.7in]{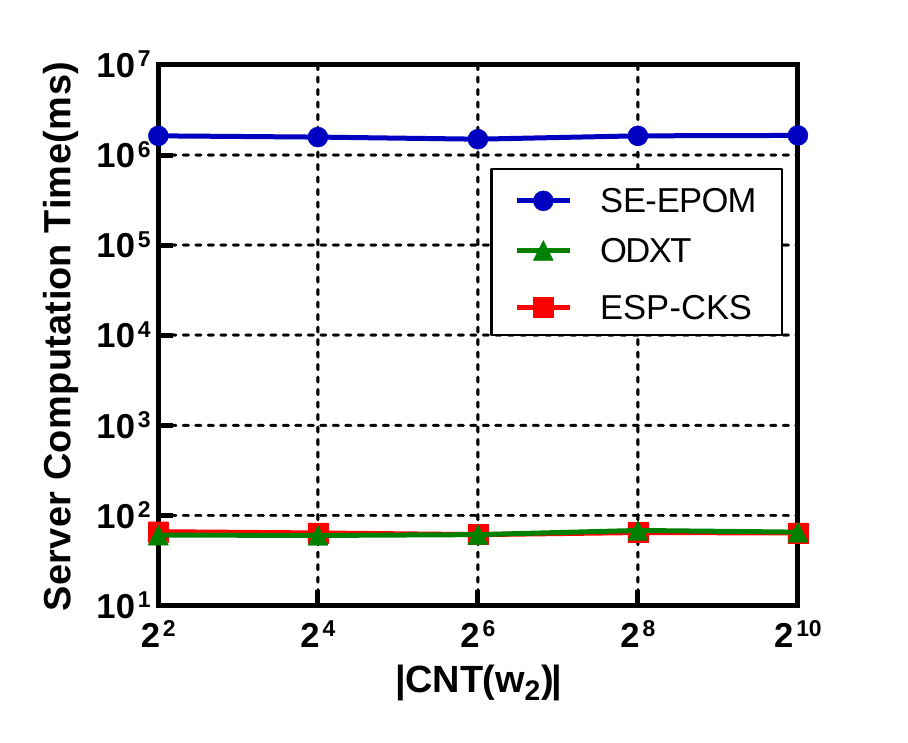}\label{Fig_3_4}}
	\hfil
	\caption{Multi-conjunctive search query $q=(w_1\wedge\cdots\wedge{w_5})$ computation time over client and server.}
	\label{Fig_3}
\end{figure}

\noindent \textbf{Setup:} Table \ref{overhead} compares the costs of the Setup (KeyGen in SE-EPOM) algorithm. The time taken in ODXT and our ESP-CKS scheme are similar and better than that of SE-EPOM, since SE-EPOM consumes more time assigning the secret keys under the public key system.

\noindent \textbf{Update:} Table \ref{overhead} also compares the cost of the client in the Update(Store in SE-EPOM) algorithm, where each update involves 1000 or 10,000 documents. For SE-EPOM, although it takes less time than our ESP-CKS, the variety of keywords is limited because it must be agreed upon in advance. Moreover, update calculation overhead of SE-EPOM scales with value $\theta$, so it is more expensive than our scheme if $\theta$ is initially set to be large and some documents have just a few keywords. Comparing ODXT with our scheme, our scheme takes longer due to keyword pair result privacy, this is because the update phase requires the Bloom filter to be encrypted, and the computational overhead of this part depends on the size of the Bloom filter. However, when we update batch files, the average cost of each file between the two schemes is not significant, because files in a batch update amortize the cost of encrypted bloom filter. Therefore, our scheme does not sacrifice significant update efficiency but offers a higher level of security.

\noindent \textbf{Search:} Figures \ref{Fig_2}, \ref{Fig_3}, and \ref{Fig_4} show the comparisons of computation overheads in Search (Test in SE-EPOM) algorithm for various schemes. 

Figure~\ref{Fig_2} shows the impact of the update frequency for conjunctive searches involving two keywords on query efficiency. In this implementation, we set the frequency of one term to a constant value and changed the frequency of another term from $2^2$ to $2^{10}$. The dataset contains 1200 files, and $\theta$ of scheme SE-EPOM is 10. Assuming $w_1$ is the keyword $s$-term with the least frequency, both the client and the server computation overhead in ESP-CKS and ODXT scale with the frequency of the $s$-term, while SE-EPOM has a constant but prohibitively expensive computation cost. This also demonstrates the significance of selecting the least frequent keyword. 

Similarly, Figure~\ref{Fig_3} shows the impact of the update frequency for conjunctive searches involving multiple keywords on query efficiency. Note that the difference in performance between ESP-CKS and ODXT is not significant, but keyword pair result privacy is protected in ESP-CKS. Therefore, our scheme is better than ODXT and SE-EPOM.

\begin{figure}[!t]
	\centerline{
		\subfigure[Client overhead]{\includegraphics[width=1.7in]{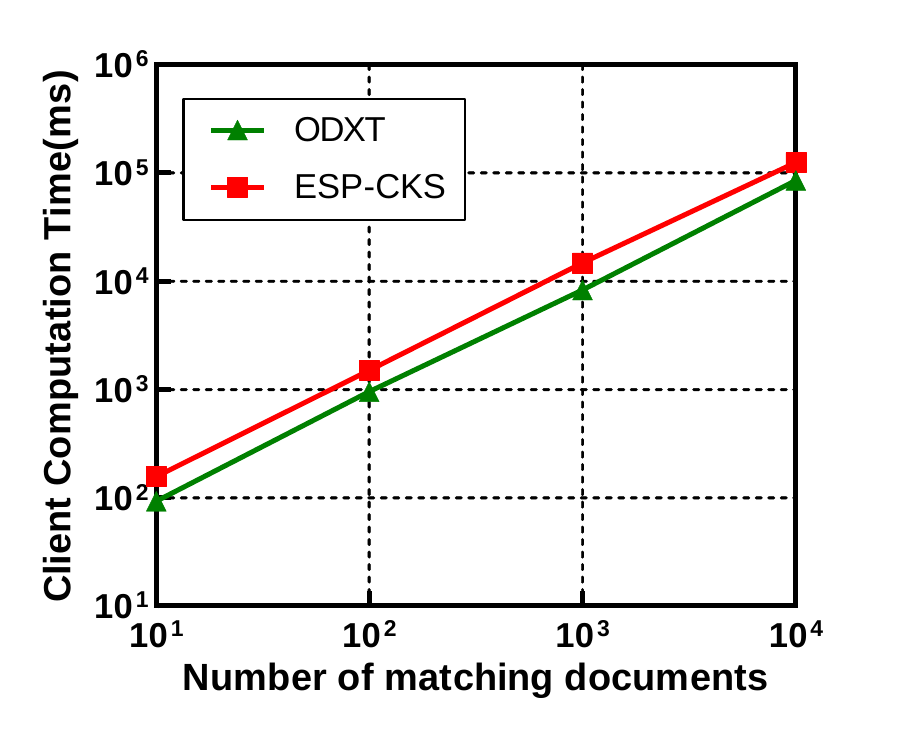}\label{Fig_4_1}}
		\hfil
		\subfigure[Server overhead]{\includegraphics[width=1.7in]{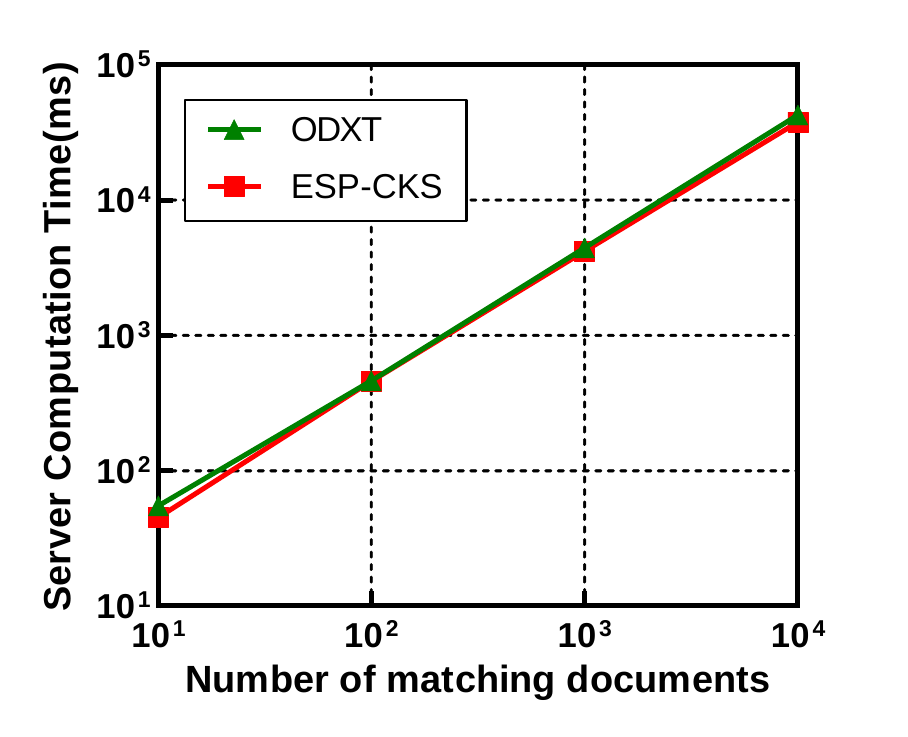}\label{Fig_4_2}}
		\hfil
	}
	\caption{Two-conjunctive search query $q=(w_1\wedge{w_2})$ computation time over client and server.}
	\label{Fig_4}
\end{figure}

Figure~\ref{Fig_4} compares ESP-CKS and ODXT for the computation overhead. We adopt the dataset with 12,500 documents and perform conjunctive search queries with 10 to 10,000 matching documents. Note that the time spent by client and server in our scheme and ODXT is proportional to the number of matching documents, and there is not much difference between the time cost of both. However, our scheme offers a higher level of security.

Based on the evaluations above, the computation overhead of SE-EPOM is prohibitively expensive, and ESP-CKS closely matches ODXT but performs better in privacy protection. Overall, our scheme delivers superior functionality and security, as shown in Table~\ref{tab1}.

\section{Related Work} \label{s_7}

SSE was first proposed by Song et al. \cite{song2000}, the search cost of the proposed SSE scheme is linear with the number of file-keyword pairs. To improve search efficiency, Curtmola et al. \cite{Curtmola2006} presented an inverted index data structure later to achieve sub-linear search time, but focused mainly on single keyword search. Cash et al. \cite{cash2013highly} proposed the first sub-linear SSE protocol named Oblivious Cross-Tags (OXT) that supports conjunctive keyword search.  The computational overhead of search scales with the update frequency for $s$-term, and sub-linear with the size of the database. It also uses a dictionary $TSet$ to make it possible to correlate a list of fixed-sized data tuples with each keyword in the database, and later allows the server to retrieve these lists via keyword-related tokens. Based on the above, the cloud server can first search for the encrypted data tuples associated with documents containing the $s$-term, then determine over encrypted data whether the documents match the other desired keywords. Lai et al. \cite{lai2018result} improved the OXT protocol in terms of private information leakage and proposed the Hidden Cross-Tags (HXT) protocol. According to \cite{lai2018result}, the OXT protocol leads to Keyword Pair Result Pattern (KPRP) leakage during the search phase. However, the file-injection attacks \cite{zhang2016all} can exploits this leakage to reveal all keywords in the conjunction with 100\% accuracy. Hence, it is essential to eliminate KPRP leakage.

The SSE schemes above focused mainly on static database settings. To support updates of the encrypted database, Kamara et al. \cite{Kamara2012} introduced the DSSE scheme, but it is vulnerable to file-injection attacks. SSE providing forward privacy firstly defined by Stefanov et al. \cite{stefanov2013practical} which enables to resist the attacks, and the first efficient DSSE scheme is proposed to achieve small leakage based on oblivious RAM (ORAM). In the setting of single keyword search, the schemes \cite{CRYPTO2016TWORAM, demertzis2019dynamic} can provide forward privacy based on ORAM. However, they suffer from prohibitively expensive computation overhead. Yang et al. \cite{yang2017rspp} proposed  a novel dynamic searchable symmetric encryption scheme with forward privacy by maintaining an increasing counter for each keyword. Song et al. \cite{song2018forward} presented the efficient FAST scheme based on symmetric key cryptography, which achieves forward privacy by singly linked lists structure and a pseudorandom permutation. Li et al. \cite{li2019searchable} developed the hidden pointer technique (HPT) to further strengthen security while ensuring forward privacy. Whereas, aforementioned works focused mainly on single keyword search. In order to provide conjunctive keyword search, Kamara and Moataz \cite{kamara2017boolean} presented IEX-2Lev and IEX-ZMF, achieving forward security. Wu et al.'s scheme \cite{wu2019vbtree} exploited counters and building trees with efficient queries to guarantee forward private conjunctive search. 

Backward privacy was later formalized by Bost et al. \cite{bost2017forward}. Subsequently, Javad et al. \cite{ghareh2018new} achieve backward privacy based on ORAM, and Sun et al. \cite{sun2018practical} adopted symmetric puncturable encryption technology to guarantee backward privacy. He et al. \cite{he2020secure} proposed CLOSE-FB scheme with forward and backward privacy based on fish-bone chain, and the scheme only stores a secret key and a global counter on the client achieving constant client storage cost. Xu et al. \cite{xu2021bestie} presented a forward and backward secure DSSE scheme named Bestie which also attains non-interactive real deletion, and Xu et al. \cite{xu2022rose} presented the robust DSSE scheme with forward and backward security based on the key-updatable pseudorandom function. However, these schemes can only perform single keyword search. Zuo et al. \cite{zuo2022range} considered range queries, and developed a scheme called FBDSSE-RQ based on their refined binary tree which can retrieve files containing keywords in a given range. For conjunctive queries that achieve forward and backward privacy, a FBDSSE-CQ scheme using extended bitmap indexing was proposed by Zuo et al. \cite{zuo2020iacr}, but it is severely limited to to the number of search keywords. The most effective existing scheme that is not restricted by keywords is the ODXT scheme \cite{patranabis2021forward,patranabis2020forward}, but there is still the threat of KPRP leakage, and the problem of frequent interactions between clients in the search phase.

\section{Conclusion}\label{s_8}
This work proposes the first efficient conjunctive keyword search ESP-CKS scheme with strong privacy. It supports dynamic databases and achieves forward and backward privacy. Besides, it eliminates the threat of keyword pair result leakage, which has never been guaranteed in existing dynamic conjunctive keyword search schemes. It is independent of the total number of documents in the database but scales with the least frequency of updates. Furthermore, the least frequent keyword acquisition protocol presented can efficiently attain the lowest frequency in a non-interactive approach, eliminating frequent communication between the data owner and search users. Compared with other keyword search approaches, our scheme offers a higher level of security, and experiment results demonstrate that it performs better.

\section*{Acknowledgments}
This work is supported by the National Natural Science Foundation of China (Grant Nos. 61972037, 61872041, U1804263), the China Postdoctoral Science Foundation (Grant Nos. 2021M700435, 2021TQ0042).

\bibliographystyle{IEEEtran}
\bibliography{fullfinal}

\appendices
\section{Proof of Theorem 1} 
We prove through a sequence of games between a challenger and an adversary. We start from $\textbf{Real}^{\Pi}_{\mathcal{A}}(\lambda)$ and construct a sequence of games where each game is designed slightly differently from the previous one, and show that they are computationally indistinguishable for the adversary $\mathcal{A}$. 
The final construction is a simulation game $\textbf{Ideal}^{\Pi}_{\mathcal{A,S}}(\lambda)$. Due to the transitive property of each two successive games' indistinguishability, $\textbf{Ideal}^{\Pi}_{\mathcal{A,S}}(\lambda)$ is computationally indistinguishable from the $\textbf{Real}^{\Pi}_{\mathcal{A}}(\lambda)$. Eventually, we conclude that the views of $\mathcal{A}$ in $\textbf{Real}^{\Pi}_{\mathcal{A}}(\lambda)$ and $\textbf{Ideal}^{\Pi}_{\mathcal{A,S}}(\lambda)$ are computationally indistinguishable, thus completing the proof of the Theorem 1. 

$\textbf{Game}_0$: Game$_0$ is the same as the real game $\textbf{Real}^{\Pi}_{\mathcal{A}}(\lambda)$ of our scheme.

$\textbf{Game}_1$: Game$_1$ is the same as Game$_0$, except that we replace the PRFs $F(K_T,\cdot)$, $F_1(K,\cdot)$, $F_2(K,\cdot)$, $F_p(K_X,\cdot)$, $F_p(K_Y,\cdot)$, and  $F_p(K_Z,\cdot)$ with random functions $G_T(\cdot)$, $G_1(\cdot)$, $G_2(\cdot)$, $G_X(\cdot)$, $G_Y(\cdot)$, and $G_Z(\cdot)$, respectively. Specifically, the function $G_T(\cdot)$ is 
uniformly randomly sampled from the set of all functions that map $\lambda$-bit string onto $\lambda$-bit string, $G_1(\cdot)$ is uniformly randomly sampled from the set of all functions that map $\lambda$-bit string onto $2\lambda$-bit string, $G_2(\cdot)$ is uniformly randomly sampled from the set of all functions that map $\lambda$-bit string onto $\lambda$-bit string, while $G_X(\cdot)$, $G_Y(\cdot)$ and $G_Z(\cdot)$ are  uniformly randomly sampled from the set of all functions that map $\lambda$-bit string onto components in $\mathbb{Z}_p^{\ast}$.

\begin{lemma}
	Suppose that $F, F_1, F_2$, and $F_p$ are secure PRFs, the views of $\mathcal{A}$ in Game$_1$ and Game$_0$ are computationally indistinguishable.
\end{lemma}

\begin{proof}
	Assuming that there exists a probabilistic polynomial-time adversary $\mathcal{B}_1$ that can distinguish between the views of adversary $\mathcal{A}$ in Game$_0$ and Game$_1$. This implies that $\mathcal{B}_1$ can be used to construct a probabilistic polynomial-time adversary $\mathcal{B}_1^{\prime}$ that can distinguish pseudorandom functions from random functions. Obviously, this is contrary to the pseudorandomness of the pseudorandom function.
\end{proof} 

$\textbf{Game}_2$: Game$_2$ is the same as Game$_1$, except that we regenerate $xtoken$ in the search phase. Specifically, for a query $q=(w_1\wedge{w_2}\cdots\wedge{w_n})$, the challenger initially looks for the adversary $\mathcal{A}$'s history of update queries to get the set of update operations $(op_j, id_j, w_1)$ involving the $s$-term $w_1$. Then, for each $x$-term $w_i$, it computes $\alpha_{i,j}$ and $xtag_{i,j}$, where $xtag_{i,j}$ is stored in an array A. The array A is used to generate $BF$ and the element $A[op,id,w]$ is added to $BF$. Next, it computes $xtoken_{i,j}=xtag_{i,j}^{1/\alpha_{i,j}}.$ For an impossible matched $xtoken$, the challenger computes $xtoken_{i,j}=g^{G_X(w_i)\cdot{G_Z(w_1||j)}}$ and stores it in an array B.

It's trivial to see that the distribution of each $xtoken$ value in Game$_2$ is the same as that of in Game$_1$. Therefore, the view of the adversary $\mathcal{A}$ in Game$_2$ is indistinguishable from the view of the adversary $\mathcal{A}$ in Game$_1$.

$\textbf{Game}_3$: Game$_3$ is the same as Game$_2$, except that we generate $\alpha$ in an alternative but equivalent way during each update phase. We replace computing $\alpha$ with sampling randomly $\alpha\stackrel{R}{\leftarrow}\mathbb{Z}_p^{\ast}$. 

\begin{lemma}
	The views of $\mathcal{A}$ in Game$_3$ and Game$_2$ are statistically indistinguishable.
\end{lemma}

\begin{proof}
	Note that the previous $\alpha$ is the product of an evaluation of $G_Y(\cdot)$ with the inverse of an evaluation of $G_Z(\cdot)$, where $G_Y(\cdot)$ and $G_Z(\cdot)$ are uniformly randomly sampled from the set of all functions that map $\lambda$-bit string onto components in $\mathbb{Z}_p^{\ast}$. Moreover, the value of $\alpha$ is also uniform and independent of the rest of the randomness in {Game}$_2$. Consequently, the distribution of each $\alpha$ value in Game$_3$ is statistically indistinguishable from that in Game$_2$.  
\end{proof}

$\textbf{Game}_4$: Game$_4$ is the same as Game$_3$, except that we regenerate $xtag$ in array A and $xtoken$ in array B. We replace $xtag$ in array A with random sample $A[op,id,w]=g^{\gamma}$, where $\gamma\stackrel{R}{\leftarrow}\mathbb{Z}_p^{\ast}$. Moreover, we replace $xtoken$ in array B with random sample $B[w_1,w_i,j]=g^{\iota}$, where $\iota\stackrel{R}{\leftarrow}\mathbb{Z}_p^{\ast}$. That is, all of the values in arrays A and B are randomly picked from $\mathbb{G}$.

\begin{lemma}
	\label{lemma3}
	Suppose that the decisional Diffie-Hellman (DDH) assumption holds in the group $\mathbb{G}$, the views of $\mathcal{A}$ in Game$_4$ and Game$_3$ are computationally indistinguishable. 
\end{lemma}

We prove Lemma \ref{lemma3} by the equivalent Lemma \ref{lemma4}, where the equivalence is derived from the (polynomial) equivalence of the DDH assumption and the extended DDH assumption [43] over any group $\mathbb{G}$. (The papaer [43] has already proved that the extended DDH assumption holds over a group $\mathbb{G}$ iff the DDH assumption holds over the same group $\mathbb{G}$.)

\begin{lemma}
	\label{lemma4}
	Suppose that the extended DDH assumption holds in the group $\mathbb{G}$, the views of $\mathcal{A}$ in Game$_4$ and Game$_3$ are computationally indistinguishable. 
\end{lemma}

\begin{proof} 
	Assuming that there exists a probabilistic polynomial-time adversary $\mathcal{B}_2$ that can distinguish between the views of adversary $\mathcal{A}$ in Game$_4$ and Game$_3$. $\mathcal{B}_2$ can easily be used to construct a probabilistic polynomial-time adversary $\mathcal{B}_2^{\prime}$ that can distinguish between $g^{G_X(w)\cdot{G_Y(id||op)}}$ and $g^{\gamma}$. Similarly, $xtoken$ in array B is the same. Obviously, $\mathcal{B}_2^{\prime}$ breaks the extended DDH assumption over the group $\mathbb{G}$, hence there is no such aforementioned probabilistic polynomial-time adversary $\mathcal{B}_2$.
\end{proof}

$\textbf{Game}_5$: Game$_5$ is the same as Game$_4$, except that we regenerate $addr$ and $val$ in the update phase. The function evaluation of the form $G_T(w||cnt||b)(b\in{\{0,1\}})$ is replaced with a function evaluation of the form $G_T(t)$, where $t$ is the timestamp when the update operation is executed. Similarly, we regenerate $saddr$ in the search phase, namely a function evaluation of the form $G_T(t)$ is substituted for the function evaluation of the form $G_T(w||cnt||0)$, where $t$ is the timestamp associated with the corresponding update operation.                                   

\begin{lemma}
	The views of $\mathcal{A}$ in Game$_5$ and Game$_4$ are computationally indistinguishable.
\end{lemma} 

\begin{proof}
	Note that the counter increases monotonically, and the uniform random function $G_T$ never takes a value on the same input at two different timestamps. Therefore, the values of $G_T$ in Game$_5$ and Game$_4$ are computationally indistinguishable from the point of view of the adversary $\mathcal{A}$.
\end{proof}

$\textbf{Game}_6$: Game$_6$ is the same as Game$_5$, except that we replace the challenger with a simulator $\mathcal{S}$ that could only access the leakage function $\mathcal{L}=(\mathcal{L}_{Setup},\mathcal{L}_{Update},\mathcal{L}_{Search})$ for each update and search query.

\begin{lemma}
	\label{lemma6}
	The views of $\mathcal{A}$ in Game$_6$ and Game$_5$ are computationally indistinguishable.
\end{lemma}

\begin{proof}
	The Simulator $\mathcal{S}$ can access an empty update leakage function $\mathcal{L}_{Update}(op,id,w)=\perp$ and search leakage function $\mathcal{L}_{Search(q)}=(\mathsf{TimeDB}(q),\mathsf{Upd}(q))$, ensuring that $\mathcal{S}$ does not have access to the actual queries issued by $\mathcal{A}$. The rest of the variables are generated by $\mathcal{S}$ as done by the challenger in Game$_5$. 
	For the search phase, $\mathcal{S}$ can learn the update frequency involving the $s$-term and the timestamp of each operation, which is expressed as $\mathsf{Upd}(q)$. It can also learn the final result of conjunctive search queries, that is, a set of document identifiers together with operation timestamps expressed as $\mathsf{TimeDB}(q)$. In addition, $\mathcal{S}$ can infer that the two conjunctive search queries $q_1$ and $q_2$ have the same $s$-term, which is subsumed by $\mathsf{Upd}(q_1)$ and $\mathsf{Upd}(q_2)$. Note that from the perspective of adversary $\mathcal{A}$, the transcripts generated by $\mathcal{S}$ are identical to the corresponding  transcripts generated by the challenger in Game$_5$.
\end{proof}

The proof of Lemma \ref{lemma6} is complete, that is, Theorem 1 is proved. 
\end{document}